\def\arxiv{1}

\ifdefined\arxiv
    \documentclass[sigconf]{acmart}  %
	\usepackage[arxiv]{optional}
	\usepackage{balance}

	\def\BibTeX{{\rm B\kern-.05em{\sc i\kern-.025em b}\kern-.08emT\kern-.1667em\lower.7ex\hbox{E}\kern-.125emX}}
\else
	\documentclass[sigconf]{acmart}  %
	\usepackage[kdd]{optional}
	\usepackage{balance}

	\def\BibTeX{{\rm B\kern-.05em{\sc i\kern-.025em b}\kern-.08emT\kern-.1667em\lower.7ex\hbox{E}\kern-.125emX}}
	
\fi

\opt{kdd,arxiv}{
\copyrightyear{2019} 
\acmYear{2019} 
\setcopyright{rightsretained}
\acmConference[KDD '19] {The 25th ACM SIGKDD Conference on Knowledge Discovery and Data Mining}{August 4--8, 2019}{Anchorage, AK, USA}
\acmBooktitle{The 25th ACM SIGKDD Conference on Knowledge Discovery and Data Mining (KDD '19), August 4--8, 2019, Anchorage, AK, USA}
\acmPrice{} 
\acmDOI{10.1145/3292500.3330674}
\acmISBN{978-1-4503-6201-6/19/08} 
\acmSubmissionID{ads0429p}
}
\opt{oldarxiv}{
\acmConference[ ]{ }
}

\RequirePackage{natbib}
\usepackage[colorinlistoftodos]{todonotes}
\usepackage{algorithm}
\usepackage[noend]{algorithmic} %
\newcommand{\algorithmicinput}{\textbf{input}}
\newcommand{\INPUT}{\item[\algorithmicinput]}
\newcommand{\algorithmicoutput}{\textbf{output}}
\newcommand{\OUTPUT}{\item[\algorithmicoutput]}

\usepackage{microtype}
\usepackage{booktabs} %
\usepackage{array}
\usepackage{xr}
\usepackage{mathtools}
\usepackage{amsthm,amssymb}
\usepackage{amsmath}
\usepackage{comment}
\usepackage{enumitem}
\usepackage{wrapfig}
\usepackage{framed}
\usepackage{subcaption}
\usepackage[capitalize]{cleveref}  
\crefname{nlem}{Lemma}{Lemmas}
\crefname{nprop}{Proposition}{Propositions}
\crefname{ncor}{Corollary}{Corollaries}
\crefname{nthm}{Theorem}{Theorems}
\crefname{exa}{Example}{Examples}
\crefname{assumption}{Assumption}{Assumptions}
\crefname{equation}{}{}
\usepackage{mathrsfs}
\usepackage{autonum} 
\usepackage{nicefrac}
\usepackage{afterpage}
\usepackage{xspace}
\usepackage{multicol}
\usepackage{enumitem}
\graphicspath{ {figures/} }
\allowdisplaybreaks[3]

\def\balign#1\ealign{\begin{align}#1\end{align}}
\def\baligns#1\ealigns{\begin{align*}#1\end{align*}}
\def\balignat#1\ealign{\begin{alignat}#1\end{alignat}}
\def\balignats#1\ealigns{\begin{alignat*}#1\end{alignat*}}
\def\bitemize#1\eitemize{\begin{itemize}#1\end{itemize}}
\def\benumerate#1\eenumerate{\begin{enumerate}#1\end{enumerate}}

\newenvironment{talign*}
 {\csname align*\endcsname}
 {\endalign}
\newenvironment{talign}
 {\csname align\endcsname}
 {\endalign}

\def\balignst#1\ealignst{\begin{talign*}#1\end{talign*}}
\def\balignt#1\ealignt{\begin{talign}#1\end{talign}}%

\let\originalleft\left
\let\originalright\right
\renewcommand{\left}{\mathopen{}\mathclose\bgroup\originalleft}
\renewcommand{\right}{\aftergroup\egroup\originalright}

\def\mbi#1{\boldsymbol{#1}} %
\def\mbf#1{\mathbf{#1}}

\def\mc#1{\mathcal{#1}}

\def\reals{\mathbb{R}} %

\def\<{\left\langle} %
\def\>{\right\rangle}

\def\defeq{\triangleq} %
\def\bs{\backslash} %

\def\norm#1{\left\|{#1}\right\|} %
\newcommand{\twonorm}[1]{\norm{#1}_2} %
\newcommand{\inner}[2]{\langle{#1},{#2}\rangle} %
\providecommand{\argmax}{\mathop\mathrm{arg max}} %
\providecommand{\argmin}{\mathop\mathrm{arg min}}

\providecommand{\sign}{\mathop\mathrm{sign}}
\ifdefined\nonewproofenvironments\else
\ifdefined\ispres\else
\newtheorem{theorem}{Theorem}

\renewenvironment{proof}{\noindent\textbf{Proof}\hspace*{1em}}{\qed\\}
\newenvironment{proof-sketch}{\noindent\textbf{Proof Sketch}
  \hspace*{1em}}{\qed\bigskip\\}
\newenvironment{proof-idea}{\noindent\textbf{Proof Idea}
  \hspace*{1em}}{\qed\bigskip\\}
\newenvironment{proof-of-lemma}[1][{}]{\noindent\textbf{Proof of Lemma {#1}}
  \hspace*{1em}}{\qed\\}
\newenvironment{proof-of-theorem}[1][{}]{\noindent\textbf{Proof of Theorem {#1}}
  \hspace*{1em}}{\qed\\}
\newenvironment{proof-attempt}{\noindent\textbf{Proof Attempt}
  \hspace*{1em}}{\qed\bigskip\\}

\fi

\newtheorem{proposition}[theorem]{Proposition}

\fi
\newcommand{\ttt}[1]{\texttt{#1}}
\newcommand{\gt}{\mathbf{y}}
\newcommand{\predgt}{\hat{\mathbf{y}}}
\newcommand{\clim}{\mathbf{c}}
\newcommand{\anom}{\mathbf{a}}
\newcommand{\predanom}{\hat{\mathbf{a}}}

\newcommand{\skill}{\text{skill}}

\newcommand{\doy}{\texttt{doy}} %
\newcommand{\monthday}{\texttt{monthday}} %

\newcommand{\trainset}{\mc{T}}

\newcommand{\US}{U.S.\xspace}
\newcommand{\dataset}{{\texttt{SubseasonalRodeo} dataset}\xspace}
\newcommand{\Dataset}{{\texttt{SubseasonalRodeo} Dataset}\xspace}
\newcommand{\autoknn}{\ttt{AutoKNN}\xspace}
\newcommand{\stepwise}{\ttt{MultiLLR}\xspace}

\newcommand{\featset}{\mc{F}}
\newcommand{\score}{\ttt{v}}

\newcolumntype{R}[1]{>{\raggedleft\let\newline\\\arraybackslash\hspace{0pt}}m{#1}}

\makeatletter
\newcommand{\subalign}[1]{%
  \vcenter{%
    \Let@ \restore@math@cr \default@tag
    \baselineskip\fontdimen10 \scriptfont\tw@
    \advance\baselineskip\fontdimen12 \scriptfont\tw@
    \lineskip\thr@@\fontdimen8 \scriptfont\thr@@
    \lineskiplimit\lineskip
    \ialign{\hfil$\m@th\scriptstyle##$&$\m@th\scriptstyle{}##$\crcr
      #1\crcr
    }%
  }
}
\makeatother

\newcommand{\ourabstract}{
\begin{abstract}
Water managers in the western United States (U.S.) rely on longterm forecasts of temperature and precipitation to prepare for droughts and other wet weather extremes.
To improve the accuracy of these longterm forecasts, the \US Bureau of Reclamation and the National Oceanic and Atmospheric Administration (NOAA) launched the Subseasonal Climate Forecast Rodeo, a year-long real-time forecasting challenge in which participants aimed to skillfully predict temperature and precipitation in the western U.S.\ two to four weeks and four to six weeks in advance.
Here we present and evaluate our machine learning approach to the Rodeo and release our \dataset, collected to train and evaluate our forecasting system.

Our system is an ensemble of two nonlinear regression models.
The first integrates the diverse collection of meteorological measurements and dynamic model forecasts in the \dataset and prunes irrelevant predictors using a customized multitask feature selection procedure.
The second uses only historical measurements of the target variable (temperature or precipitation) and introduces multitask nearest neighbor features into a weighted local linear regression.
Each model alone is significantly more accurate than the debiased operational U.S. Climate Forecasting System (CFSv2), and our ensemble skill exceeds that of the top Rodeo competitor for each target variable and forecast horizon.
Moreover, over 2011-2018, an ensemble of our regression models and debiased CFSv2
improves debiased CFSv2 skill by 40-50\% for temperature and 129-169\% for precipitation.
We hope that both our dataset and our methods will help to advance the state of the art in subseasonal forecasting.
\end{abstract}}

\opt{kdd,arxiv}{
\setcounter{secnumdepth}{2} 

\title[Subseasonal Forecasting with Machine Learning]{Improving Subseasonal Forecasting in the Western U.S.\@ with Machine Learning}

\author{Jessica Hwang}
\affiliation{%
  \institution{Department of Statistics, Stanford University}
  \city{Stanford}
  \state{California}
}
\email{jjhwang@stanford.edu}
\author{Paulo Orenstein}
\affiliation{%
  \institution{Department of Statistics, Stanford University}
  \city{Stanford}
  \state{California}
}
\email{pauloo@stanford.edu}
\author{Judah Cohen}
\affiliation{%
  \institution{Atmospheric and Environmental Research}
  \city{Lexington}
  \state{MA}
}
\email{jcohen@aer.com}
\author{Karl Pfeiffer}
\affiliation{%
  \institution{Atmospheric and Environmental Research}
  \city{Lexington}
  \state{MA}
}
\email{kpfeiffe@aer.com}
\author{Lester Mackey}
\affiliation{%
  \institution{Microsoft Research New England}
  \city{Cambridge}
  \state{MA}
}
\email{lmackey@microsoft.com}
\renewcommand{\shortauthors}{J.\ Hwang et al.}
\ourabstract

\ccsdesc[300]{Applied computing~Environmental sciences}

\opt{kdd,arxiv}{
\keywords{Subseasonal climate forecasting, 
Temperature, Precipitation, Multitask feature selection, Multitask KNN, Ensembling, Western United States, Drought, Water management}}
}

\begin{document}
\opt{kdd,arxiv}{\maketitle}
\opt{oldarxiv}{\newpage}

\section{Introduction}
Water and fire managers in the western United States (\US) rely on \emph{subseasonal forecasts}---forecasts of temperature and precipitation two to six weeks in advance---to allocate water resources, manage wildfires, and prepare for droughts and other weather extremes \citep{white2017}.
While purely physics-based numerical weather prediction dominates the landscape of short-term weather forecasting, such deterministic methods have a limited \emph{skillful} (i.e., accurate) forecast horizon due to the chaotic nature of their differential equations \citep{lorenz1963deterministic}. Prior to the widespread availability of operational numerical weather prediction, weather forecasters made predictions using their knowledge of past weather patterns and climate (sometimes called \emph{the method of analogs}) \citep{nebeker1995calculating}. The current availability of ample meteorological records and  high-performance computing offers the opportunity to blend physics-based and statistical machine learning (ML) approaches to extend the skillful forecast horizon.

This data and computing opportunity, coupled with the critical operational need, motivated the \US Bureau of Reclamation and the National Oceanic and Atmospheric Administration (NOAA) to conduct the Subseasonal Climate Forecast Rodeo \citep{nowak2017sub}, a year-long real-time forecasting challenge, in which participants aimed to skillfully predict temperature and precipitation in the western \US two to four weeks and four to six weeks in advance. To meet this challenge, we developed an ML-based forecasting system and a \dataset \citep{dataset2018} suitable for training and benchmarking subseasonal forecasts.

ML approaches have been successfully applied to both
short-term ($< 2$ week) weather forecasting \citep{doi:10.1175/WAF-D-17-0188.1,doi:10.1175/WAF-D-17-0006.1,doi:10.1175/MWR-D-17-0307.1, xingjian2015convolutional, hernandez2016rainfall, qiu2017short, ghaderi2017deep, kuligowski1998localized, horenko2008automated, radhika2009atmospheric, karevan2016spatio, cofino2002bayesian, grover2015deep} and longer-term climate prediction \citep{doi:10.1175/JCLI-D-15-0648.1,doi:10.1175/JAMC-D-13-0181.1,doi:10.1002/2017GL075674, iglesias2015examination,cohen2019s2s}, but mid-term subseasonal outlooks, which depend on both local weather and global climate variables, still lack skillful forecasts \citep{robertson2015}.

Our subseasonal ML system is an ensemble of two nonlinear regression models: a local linear regression model with multitask feature selection (\stepwise) and a weighted local autoregression enhanced with multitask $k$-nearest neighbor features (\autoknn). 
The \stepwise model introduces candidate regressors from each data source in the \dataset and then prunes irrelevant predictors using a multitask backward stepwise criterion designed for the forecasting skill objective. 
The {\autoknn} model extracts features only from the target variable (temperature or precipitation), combining lagged measurements with a skill-specific form of nearest-neighbor modeling.
For each of the two Rodeo target variables (temperature and precipitation) and forecast horizons (weeks 3-4 and weeks 5-6), this paper makes the following principal contributions:
\benumerate
\item We release a new \dataset suitable for training and benchmarking subseasonal forecasts.
\item We introduce two subseasonal regression approaches tailored to the forecast skill objective, one of which uses only features of the target variable.
\item We introduce a simple ensembling procedure that provably improves average skill whenever average skill is positive.
\item We show that each regression method alone outperforms the Rodeo benchmarks, including a debiased version of the operational \US Climate Forecasting System (CFSv2), and that our ensemble outperforms the top Rodeo competitor.
\item We show that, over 2011-2018, an ensemble of our models and debiased CFSv2 improves debiased CFSv2 skill by 40-50\% for temperature and 129-169\% for precipitation.
\eenumerate
We hope that this work will expose the ML community to an important problem ripe for ML development---improving subseasonal forecasting for water and fire management,
demonstrate that ML tools can lead to significant improvements in subseasonal forecasting skill, and stimulate future development with the release of our user-friendly Python Pandas \dataset.

\subsection{Related Work} \label{sec:related}
While statistical modeling was common in the early days of weather and climate forecasting \citep{nebeker1995calculating}, purely physics-based dynamical modeling of atmosphere and oceans rose to prominence in the 1980s and has been the dominant forecasting paradigm in major climate prediction centers since the 1990s \cite{barnston2012skill}.  
Nevertheless, skillful statistical machine learning approaches have been developed for short-term weather forecasting with outlooks ranging from hours to two weeks ahead \citep{doi:10.1175/WAF-D-17-0188.1,doi:10.1175/WAF-D-17-0006.1,doi:10.1175/MWR-D-17-0307.1, xingjian2015convolutional, hernandez2016rainfall, qiu2017short, ghaderi2017deep, kuligowski1998localized, horenko2008automated, radhika2009atmospheric, karevan2016spatio, cofino2002bayesian, grover2015deep} and for coarse-grained long-term climate forecasts with target variables aggregated over months or years \citep{doi:10.1175/JCLI-D-15-0648.1,doi:10.1175/JAMC-D-13-0181.1,doi:10.1002/2017GL075674, iglesias2015examination,cohen2019s2s}.
Tailored machine learning solutions are also available for detecting and predicting weather extremes \citep{mcgovern2014enhancing,liu2016application,racah2017extremeweather}.
However, subseasonal forecasting, with its 2-6 week outlooks and biweekly granularity, is considered more difficult than either short-term weather forecasting or long-term climate forecasting, due to its complex dependence on both local weather and global climate variables \cite{white2017}.
We complement prior work by developing a dataset and an ML-based forecasting system suitable for improving temperature and precipitation prediction in this traditional `predictability desert' \citep{vitart2012subseasonal}.  

\section{The Subseasonal Climate Forecast Rodeo} 
\label{sec:rodeo}
The Subseasonal Climate Forecast Rodeo was a year-long, real-time forecasting competition in which, every two weeks, contestants submitted forecasts for average temperature ($^\circ$C) and total precipitation (mm) at two forecast horizons, 15-28 days ahead (weeks 3-4) and 29-42 days ahead (weeks 5-6). The geographic region of interest was the western contiguous United States, delimited by latitudes 25N to 50N and longitudes 125W to 93W, at a 1$^\circ$ by 1$^\circ$ resolution, for a total of $G=514$ grid points. The initial forecasts were issued on April 18, 2017 and the final on April 3, 2018.

Forecasts were judged on the spatial cosine similarity between predictions and observations adjusted by a long-term average.
More precisely, let $t$ denote a date represented by the number of days since January $1$, $1901$, and let $\texttt{year}(t)$, $\doy(t)$, and $\monthday(t)$ respectively denote the year, the day of the year, and the month-day combination (e.g., January $1$) associated with that date.
We associate with the two-week period beginning on $t$ an observed average temperature or total precipitation $\gt_{t} \in \mathbf{R}^G$ and an observed \emph{anomaly}
\balignt
\anom_{t} &= \gt_{t} - \clim_{\monthday(t)}, 
\ealignt
where
\balignt
\clim_{d} 
	&\defeq \frac{1}{30} \sum_{\subalign{&t\,:\,\monthday(t)=d,\\ &1981\leq\texttt{year}(t)\leq 2010}} \gt_{t}
\ealignt
is the \emph{climatology} or long-term average over 1981-2010 for the month-day combination $d$. Contestant forecasts $\predgt_{t}$ were judged on the cosine similarity---termed \emph{skill} in meteorology---between their forecast anomalies $\predanom_{t} = \predgt_{t} - \clim_{\monthday(t)}$ and the observed anomalies:
\balignt\label{eqn:skill}
\skill(\predanom_{t}, \anom_{t}) \defeq \cos(\predanom_{t}, \anom_{t}) = \frac{\inner{\predanom_{t}}{ \anom_{t}}}{\|\predanom_{t}\|_2\|\anom_{t}\|_2}.
\ealignt

To qualify for a prize, contestants had to achieve higher mean skill over all forecasts than two benchmarks, a debiased version of the physics-based operational \US Climate Forecasting System (CFSv2) and a damped persistence forecast. The official contest CFSv2 forecast for $t$, an average of 32 operational forecasts based on 4 model initializations and 8 lead times, was debiased by adding the mean observed temperature or precipitation for $\monthday(t)$ over 1999-2010 and subtracting the mean CFSv2 reforecast, an average of 8 lead times for a single initialization, over the same period.
An exact description of the damped persistence model was not provided, but the Rodeo organizers reported it relied on ``seasonally developed regression coefficients based on the historical climatology period of 1981-2010 that relate observations of the past two weeks to the forecast outlook periods on a grid cell by grid cell basis.''

\section{Our \Dataset}
\label{sec:data}
Since the Rodeo did not provide data for training predictive models, we constructed our own \dataset from a diverse collection of data sources. Unless otherwise noted below, spatiotemporal variables were interpolated to a 1$^\circ$ by 1$^\circ$ grid and restricted to the contest grid points, and daily measurements were replaced with average measurements over the ensuing two-week period.
The \dataset is available for download at \cite{dataset2018}, and \cref{sec:data_supp} provides additional details on data sources, processing, and variables ultimately not used in our solution.

\textbf{Temperature\quad} Daily maximum and minimum temperature measurements at 2 meters (\texttt{tmax} and \ttt{tmin}) from 1979 onwards were obtained from NOAA's Climate Prediction Center (CPC) Global Gridded Temperature dataset and converted to $^\circ$C; the same data source was used to evaluate contestant forecasts. The official contest target temperature variable was $\ttt{tmp2m}\defeq \frac{\ttt{tmax} + \ttt{tmin}}{2}$.

\textbf{Precipitation\quad} 
Daily precipitation (\texttt{precip}) data from 1979 onward were obtained from NOAA's CPC Gauge-Based Analysis of Global Daily Precipitation \citep{xie2010cpc} and converted to mm; the same data source was used to evaluate contestant forecasts.
We augmented this dataset with daily \US precipitation data in mm from 1948-1979 from the CPC Unified Gauge-Based Analysis of Daily Precipitation over CONUS.
Measurements were replaced with sums over the ensuing two-week period.

\textbf{Sea surface temperature and sea ice concentration\,} NOAA's Optimum Interpolation Sea Surface Temperature (SST) dataset provides SST and sea ice concentration data, daily from 1981 to the present \citep{reynolds2007daily}. 
After interpolation, we extracted the top three principal components (PCs), $(\texttt{sst\_i})_{i=1}^3$ and $(\texttt{icec\_i})_{i=1}^3$, across grid points in the Pacific basin region (20S to 65N, 150E to 90W) based on PC loadings from 1981-2010.

\textbf{Multivariate ENSO index (MEI)\quad} Bimonthly MEI values (\texttt{mei}) from 1949 to the present, were obtained from NOAA/Earth System Research Laboratory \citep{wolter1993monitoring,wolter1998measuring}. %
The MEI is a scalar summary of six variables (sea-level pressure, zonal and meridional surface wind components, SST, surface air temperature, and sky cloudiness) associated with El Ni\~no/Southern Oscillation (ENSO), an ocean-atmosphere coupled climate mode.

\textbf{Madden-Julian oscillation (MJO)\quad} Daily MJO values since 1974 are provided by the Australian Government Bureau of Meteorology \citep{wheeler2004all}. MJO is a metric of tropical convection on daily to weekly timescales and can have significant impact on the western United States' subseasonal climate. We extract measurements of \texttt{phase} and \texttt{amplitude} on the target date but do not aggregate over the two-week period.

\textbf{Relative humidity and pressure\quad} NOAA's National Center for Environmental Prediction (NCEP)/National Center for Atmospheric Research Reanalysis dataset \citep{kalnay1996ncep} contains daily relative humidity  (\texttt{rhum}) near the surface (sigma level 0.995) from 1948 to the present and daily pressure at the surface (\texttt{pres}) from 1979 to the present.

\textbf{Geopotential height\quad} To capture polar vortex variability, we obtained daily mean geopotential height at 10mb since 1948 from the NCEP Reanalysis dataset \citep{kalnay1996ncep} and extracted the top three PCs $(\texttt{wind\_hgt\_10\_i})_{i=1}^3$ based on PC loadings from 1948-2010. No interpolation or contest grid restriction was performed.

\textbf{NMME\quad} The North American Multi-Model Ensemble (NMME) is a collection of physics-based forecast models from various modeling centers in North America \citep{kirtman2014north}. Forecasts issued monthly from the Cansips, CanCM3, CanCM4, CCSM3, CCSM4, GFDL-CM2.1-aer04, GFDL-CM2.5 FLOR-A06 and FLOR-B01, NASA-GMAO-062012, and NCEP-CFSv2 models were downloaded from the IRI/LDEO Climate Data Library. Each forecast contains monthly mean predictions from 0.5 to 8.5 months ahead. We derived forecasts by taking a weighted average of the monthly predictions with weights proportional to the number of target period days that fell into each month. We then formed an equally-weighted average (\texttt{nmme\_wo\_ccsm3\_nasa}) of all models save CCSM3 and NASA (which were not reliably updated during the contest). 
Another feature was created by averaging the most recent monthly forecast of each model save CCSM3 and NASA (\texttt{nmme0\_wo\_ccsm3\_nasa}).
\section{Forecasting Models}
\label{sec:methods}
In developing our forecasting models, we focused our attention on computationally efficient methods that exploited the \emph{multitask}, i.e., multiple grid point, nature of our problem and incorporated the unusual forecasting skill objective function \cref{eqn:skill}.
For each target variable (temperature or precipitation) and horizon (weeks 3-4 or 5-6), our forecasting system relies on two regression models trained using two sets of features derived from the \dataset. 
The first model, described in \cref{sec:stepwise}, introduces lagged measurements from all data sources in the \dataset as candidate regressors. For each target date, irrelevant regressors are pruned automatically using multitask feature selection tailored to the cosine similarity objective.
Our second model, described in \cref{sec:knn}, chooses features derived from the target variable (temperature or precipitation) using a skill-specific nearest neighbor strategy.
The final forecast is obtained by ensembling the predictions of these two models in a manner well-suited to the cosine similarity objective.
\subsection{Local Linear Regression with Multitask Feature Selection (\stepwise)}
\label{sec:stepwise}
Our first model uses lagged measurements from each of the data sources in the \dataset as candidate regressors, with lags selected based on the temporal resolution of the measurement and the frequency of the data source update.
The y-axis of \cref{fig:stepwise_features} provides an explicit list of candidate regressors for each prediction task. 
The suffix $\ttt{anom}$ indicates that feature values are anomalies instead of raw measurements, the substring $\ttt{shift}\ell$ indicates a lagged feature with measurements from $\ell$ days prior, and the constant feature \ttt{ones} equals $1$ for all datapoints.

We combine predictors using local linear regression with locality determined by the day of the year\footnote{As a matter of convention, we treat Feb. 29 as the same day as Feb. 28 when  computing $\doy$, so that $\doy(t) \in \{1, \dots, 365\}$.} (\cref{alg:wllr}).
Specifically, the training data for a given target date is restricted to a 56-day (8-week) span around the target date's day of the year ($s = 56$). For example, if the target date is May 2, 2017, the training data consists of days within 56 days of May 2 in any year.  We employ equal datapoint weighting ($w_{t,g} = 1$) and no offsets ($b_{t,g} = 0$).

\begin{algorithm}[tb]
  \caption{Weighted Local Linear Regression}
  \label{alg:wllr}
  \begin{algorithmic}
  	\INPUT test day of year $d^*$; span $s$;
    training outcomes, features, offsets, and weights $(y_{t,g}, \mbf{x}_{t,g}, b_{t,g}, w_{t,g})_{t\in\trainset, g\in \{1,\dots, G\}}$
	\STATE $\mc{D} \defeq \{t \in \trainset:\frac{365}{2} - ||\doy(t)-d^*| -\frac{365}{2}| \leq s\}$
    \FOR{grid points $g = 1$ {\bfseries to} $G$}
    \STATE
    $%
    \mbi{\hat{\beta}}_g \in \argmin_{\mbi{\beta}} 
    \sum_{t \in \mc{D}}
    w_{t,g}(y_{t,g} - b_{t,g} - {\mbi{\beta}}^\top{\mbf{x}_{t,g}})^2
    $
    \ENDFOR
    \OUTPUT coefficients $(\mbi{\hat{\beta}}_g)_{g=1}^G$
  \end{algorithmic}
\end{algorithm}

As we do not expect all features to be relevant at all times of year, we use multitask feature selection tailored to the cosine objective to automatically identify relevant features for each target date. The selection is multitask in that variables for a target date are selected jointly for all grid points, while the coefficients associated with those variables are fit independently for each grid point using local linear regression.

The feature selection is performed for each target date using a customized backward stepwise procedure (\cref{alg:stepwise}) built atop the local linear regression subroutine. At each step of the backward stepwise procedure, we regress the outcome on all remaining candidate predictors; the regression is fit separately for each grid point. A measure of predictive performance (described in the next paragraph) is computed, and the candidate predictor that decreases predictive performance the least is removed. The procedure terminates when no candidate predictor can be removed from the model without decreasing predictive performance by more than the tolerance threshold $\ttt{tol}=0.01$.

\begin{algorithm}[tb]
  \caption{Multitask Backward Stepwise Feature Selection}
  \label{alg:stepwise}
  \begin{algorithmic}
  	\INPUT test day of year $d^*$; set of feature identifiers $\featset$;
    base regression procedure \ttt{BaseReg}; tolerance $\ttt{tol}$
    \STATE $\mc{D} \defeq \{t:\doy(t)=d^*\}$; $\texttt{converged} = \texttt{False}$
    \STATE $\score = \ttt{LOYOCV}(d^*, \ttt{BaseReg}, \featset)$
    \WHILE{\texttt{not converged}}
    \FORALL{feature identifiers  $j \in \featset$} 
    \STATE $(\predanom_t)_{t\in\mc{D}} \leftarrow \ttt{LOYOCV}(d^*, \ttt{BaseReg}, \featset \bs \{j\})$
    \STATE $\score_j = \frac{1}{|\mc{D}|}\sum_{t\in\mc{D}}\skill(\predanom_t, \anom_t)$
    \ENDFOR
    \IF{$\ttt{tol} > \score - \max_{j\in\featset}\score_j$}
    	\STATE  $j^* = \argmax_{j\in\featset}\score_j$; $\score = \score_{j^*}$; 
        $\featset = \featset\bs \{j^*\}$
    \ELSE 
    	\STATE $\ttt{converged} = \ttt{True}$
    \ENDIF
    \ENDWHILE
    \OUTPUT selected feature identifiers $\featset$
  \end{algorithmic}
\end{algorithm}

Our measure of predictive performance is the average leave-one-year-out cross-validated (\ttt{LOYOCV}) skill on the target date's day-of-year, where the average is taken across all years in the training data. The \ttt{LOYOCV} skill for a target date $t$ is the cosine similarity achieved by holding out a year's worth of data around $t$, fitting the model on the remaining data, and predicting the outcome for $t$. When forecasting weeks 3-4, we hold out the data from 29 days before $t$ through 335 days after $t$; for weeks 5-6, we hold out the data from 43 days before through 321 days after $t$. This ensures that the model is not fitted on future dates too close to $t$.
For $n$ training dates, $Y$ training years, and $d$ features, the \stepwise running time is $O(nd^2+Yd^3)$ per grid point and step. In our experiments in \cref{sec:experiments}, we run the per grid point regressions in parallel on each step, $d$ ranges from 20 to 23, and the average number of steps is $13$.

\subsection{Multitask $k$-Nearest Neighbor Autoregression (\autoknn)} \label{sec:knn}
Our second model is a weighted local linear regression (\cref{alg:wllr}) with features derived exclusively from historical measurements of the target variable (temperature or precipitation). When predicting weeks 3-4, we include lagged temperature or precipitation anomalies from 29 days, 58 days, and 1 year prior to the target date; when predicting weeks 5-6, we use 43 days, 86 days, and 1 year. These lags are chosen because the most recent data available to us are from 29 days before the target date when predicting weeks 3-4 and 58 days before the target date when predicting weeks 5-6.

In addition to fixed lags, we include the constant intercept  \texttt{ones} and the observed anomaly patterns of the target variable on similar dates in the past (\cref{alg:knn}). 
Our measure of similarity is tailored to the cosine similarity objective: similarity between a target date and another date is measured as the mean skill observed when the historical anomalies preceding the candidate date are used to forecast the historical anomalies of the target date. The mean skill is computed over a history of $H=60$ days, starting 1 year prior to the target date (lag $\ell = 365$). Only dates with observations fully observed prior to the forecast issue date are considered viable. We find the 20 viable candidate dates with the highest similarity to the target date and scale each neighbor date's observed anomaly vector so that it has a standard deviation equal to 1. 
The resulting features are \texttt{knn1} (the most similar neighbor) through \texttt{knn20} (the 20th most similar neighbor).
We find the $k = 20$ top neighbors for each of $n$ training dates in parallel, using $O(knHG)$ time per date. 

\begin{algorithm}[tb]
  \caption{Multitask $k$-Nearest Neighbor Similarities}
  \label{alg:knn}
  \begin{algorithmic}
	\INPUT test date $t^*$; training anomalies $(\anom_{t})_t$; lag $\ell$; 
    history $H$\FORALL{training dates $t$}
    \STATE $\texttt{sim}_{t} = \frac{1}{H}\sum_{h=0}^{H-1} \skill(\anom_{t-\ell-h}, \anom_{t^*-\ell-h})$
    \ENDFOR
    \OUTPUT similarities $(\texttt{sim}_{t})_t$
  \end{algorithmic}
\end{algorithm}

To predict a given target date, we regress onto the three fixed lags, the constant intercept feature \texttt{ones}, and either \texttt{knn1} through \texttt{knn20} (for temperature) or \texttt{knn1} only (for precipitation), treating each grid point as a separate prediction task. We found that including \texttt{knn2} through \texttt{knn20} did not lead to improved performance for predicting precipitation. For each grid point, we fit a weighted local linear regression, with weights $w_{t,g}$ given by $1$ over the variance of the target anomaly vector. As with \stepwise, locality is determined by the day of the year. For predicting precipitation, we restrict the training data to a 56-day span $s$ around the target date's day of the year. For predicting temperature, we use all dates.
In each case, we use a climatology offset ($b_{t,g} = c_{\monthday(t),g}$) so that the effective target variable is the measurement anomaly rather than the raw measurement.
Given $d$ features and $n$ training dates, the final regression is carried out in $O(nd^2)$ time per grid point. In our experiments in \cref{sec:experiments}, per grid point regressions were performed in parallel, and $d = 24$ for temperature and $d = 5$ for precipitation.

\subsection{Ensembling} \label{sec:ensemble}

Our final forecasting model is obtained by ensembling the predictions of the \stepwise and \autoknn models. Specifically, for a given target date, we take as our ensemble forecast anomalies the average of the $\ell_2$-normalized predicted anomalies of the two models: 
\balignt
\label{eqn:ens_stepwise_knn}
\predanom_{\textrm{ensemble}} \defeq \frac{1}{2}\frac{\predanom_{\textrm{multillr}}}{\twonorm{\predanom_\textrm{multillr}}} + \frac{1}{2}\frac{\predanom_\textrm{autoknn}}{\twonorm{\predanom_\textrm{autoknn}}}.
\ealignt
The $\ell_2$ normalization is motivated by the following result, which implies that the skill of $\predanom_{\textrm{ensemble}}$ is strictly better than the average skill of $\predanom_\textrm{multillr}$ and $\predanom_\textrm{autoknn}$ whenever that average skill is positive.

\begin{proposition}
\label{prop:ensemble}
Consider an observed anomaly vector $\anom$ and $m$ distinct forecast anomaly vectors $(\predanom_{i})_{i=1}^m$.
For any vector of weights $\mbf{p}\in \reals^m$ with $\sum_{i=1}^m p_i = 1$ and $p_i \geq 0$,
let %
\balignt
\bar{\anom}_{(\mbf{p})} \defeq \sum_{i=1}^m p_i \frac{\predanom_i}{\twonorm{\predanom_i}}
\ealignt
be the weighted average of the $\ell_2$-normalized forecast anomalies.
Then,
\balignt
\sign(\sum_{i=1}^m p_i \cos(\predanom_i, \anom)) &= \sign(\cos(\bar{\anom}_{(\mbf{p})}, \anom))
\ealignt
and
\balignt
|\sum_{i=1}^m p_i \cos(\predanom_i, \anom)| &\leq |\cos(\bar{\anom}_{(\mbf{p})}, \anom)|, 
\ealignt
with strict inequality whenever $\sum_{i=1}^m p_i \cos(\predanom_i, \anom) \neq 0$.
Hence, whenever the weighted average of individual anomaly skills is positive,
the skill of $\bar{\anom}_{(\mbf{p})}$ is strictly greater than the weighted average of the individual skills.
\end{proposition}

\begin{proof}
The sign claim follows from the equalities
\balignt
\sum_{i=1}^m p_i \cos(\predanom_i, \anom)
	&= \sum_{i=1}^m p_i \inner{\frac{\predanom_i}{\twonorm{\predanom_i}}}{\frac{\anom}{\twonorm{\anom}}} \\
	&= \textstyle\inner{\bar{\anom}_{(\mbf{p})}}{\frac{\anom}{\twonorm{\anom}}} 
    = \textstyle\cos(\bar{\anom}_{(\mbf{p})}, \anom) \  \|\bar{\anom}_{(\mbf{p})}\|_2. 
\ealignt
Since the forecasts are distinct, Jensen's inequality now yields the magnitude claim as
\balignt
\textstyle|\sum_{i=1}^m p_i \cos(\predanom_i, \anom)|
	&= |\cos(\bar{\anom}_{(\mbf{p})}, \anom)|\  \|\bar{\anom}_{(\mbf{p})}\|_2 \\
	&\leq |\cos(\bar{\anom}_{(\mbf{p})}, \anom)| \sum_{i=1}^m p_i \frac{\twonorm{\predanom_i}}{\twonorm{\predanom_i}} 
    = |\cos(\bar{\anom}_{(\mbf{p})}, \anom)|,
\ealignt
with strict inequality when $\sum_{i=1}^m p_i \cos(\predanom_i, \anom) \neq 0$.
\end{proof}

\section{Experiments}\label{sec:experiments}
In this section we evaluate our model forecasts over the Rodeo contest period and over each year following the climatology period and explore the relevant features inferred by each model.
Python 2.7 code to reproduce all experiments can be found at \url{https://github.com/paulo-o/forecast_rodeo}. 
\subsection{Contest Baselines}
For each target date in the contest period, the Rodeo organizers provided the skills of two baseline models, debiased CFSv2 and damped persistence. 
To provide baselines for evaluation outside the contest period, we reconstructed a debiased CFSv2 forecast approximating the contest guidelines.
We were unable to recreate the damped persistence model, as no exact description was provided.

We first reconstructed the undebiased 2011-2018 CFSv2 forecasts using the 6-hourly CFSv2 Operational Forecast dataset
and, for each month-day combination, computed long-term CFS reforecast averages over 1999-2010 using the $6$-hourly CFS Reforecast High-Priority Subset \citep{saha2014ncep}. For each target two-week period and horizon, we averaged eight forecasts, issued at 6-hourly intervals. For weeks 3-4, the eight forecasts came from 15 and 16 days prior to the target date; for weeks 5-6, we used 29 and 30 days prior.
For each date $t$, we then reconstructed the debiased CFSv2 forecast by subtracting the long-term CFS average and adding the observed target variable average over 1999-2010 for $\monthday(t)$ to the reconstructed CFSv2 forecast. 
Our reconstructed debiased forecasts are available for download at \cite{datasetCFSv2}, and \cref{sec:cfs_supp} provides more details on data sources and processing.

While the official contest CFSv2 baseline averages the forecasts of four model initializations, the CFSv2 Operational Forecast dataset only provides the forecasts of one model initialization (the remaining model initialization forecasts are released in real time but deleted after one week). Thus, our reconstruction does not precisely match the contest baseline, but it provides a similarly competitive benchmark.

\subsection{Contest Period Evaluation} \label{subsec:contest_period_evaluation}
We now examine how our methods perform over the contest period, consisting of forecast issue dates between April 18, 2017, and April 17, 2018. Forecast issue dates occur every two weeks, so we have 26 realized skills for each method and each prediction task. \cref{tab:contest_period_skills} shows the average skills for each of our methods and each of the baselines. All three of our methods outperform the official contest baselines (debiased CFSv2 and damped persistence), and our ensemble outperforms the top Rodeo competitor in all four prediction tasks.
Note that, while the remaining evaluations are of static modeling strategies, the competitor skills represent the real-time evaluations of forecasting systems that may have evolved over the course of the competition.

In \cref{fig:contest_year_histograms} we plot the 26 realized skills for each method. In each plot, the average skill over the contest period is indicated by a vertical line. The histograms indicate that both of the official contest baselines have a number of extreme negative skills, which drag down their average skill over the contest period. Our ensemble avoids these extreme negative skills. For both precipitation tasks, the worst realized skills of the two baseline methods are $-0.8$ or worse; by contrast, the worst realized skill of the ensemble is $-0.4$.
\begin{table*}[t!]%
\centering
\caption{Average contest-period skill of the proposed models \stepwise and \autoknn, the proposed ensemble of \stepwise and \autoknn (\emph{ensemble}), the official contest debiased-CFSv2 baseline, the official contest damped-persistence baseline (\emph{damped}), and the top-performing competitor in the Forecast Rodeo contest (\emph{top competitor}).
See \cref{subsec:contest_period_evaluation} for more details.}
\label{tab:contest_period_skills}
\begin{tabular}{l|ccccccc}
\toprule
task & multillr & autoknn & ensemble & contest debiased cfsv2 & damped & top competitor\\
\midrule
temperature, weeks 3-4 & 0.3079 & 0.2807 & \bf{0.3451} & 0.1589 & 0.1952 & 0.2855\\
temperature, weeks 5-6 & 0.2562 & 0.2817 & \bf{0.3025} & 0.2192 & \!\!-0.0762 & 0.2357\\
precipitation, weeks 3-4 & 0.1597 & 0.2156 & \bf{0.2364} & 0.0713 & \!\!-0.1463 & 0.2144\\
precipitation, weeks 5-6 & 0.1876 & 0.1870 & \bf{0.2315} & 0.0227 & \!\!-0.1613 & 0.2162\\
\bottomrule
\end{tabular}
\end{table*}
\begin{figure*}[h!]%
\centering
\includegraphics[width=.9\textwidth]{%
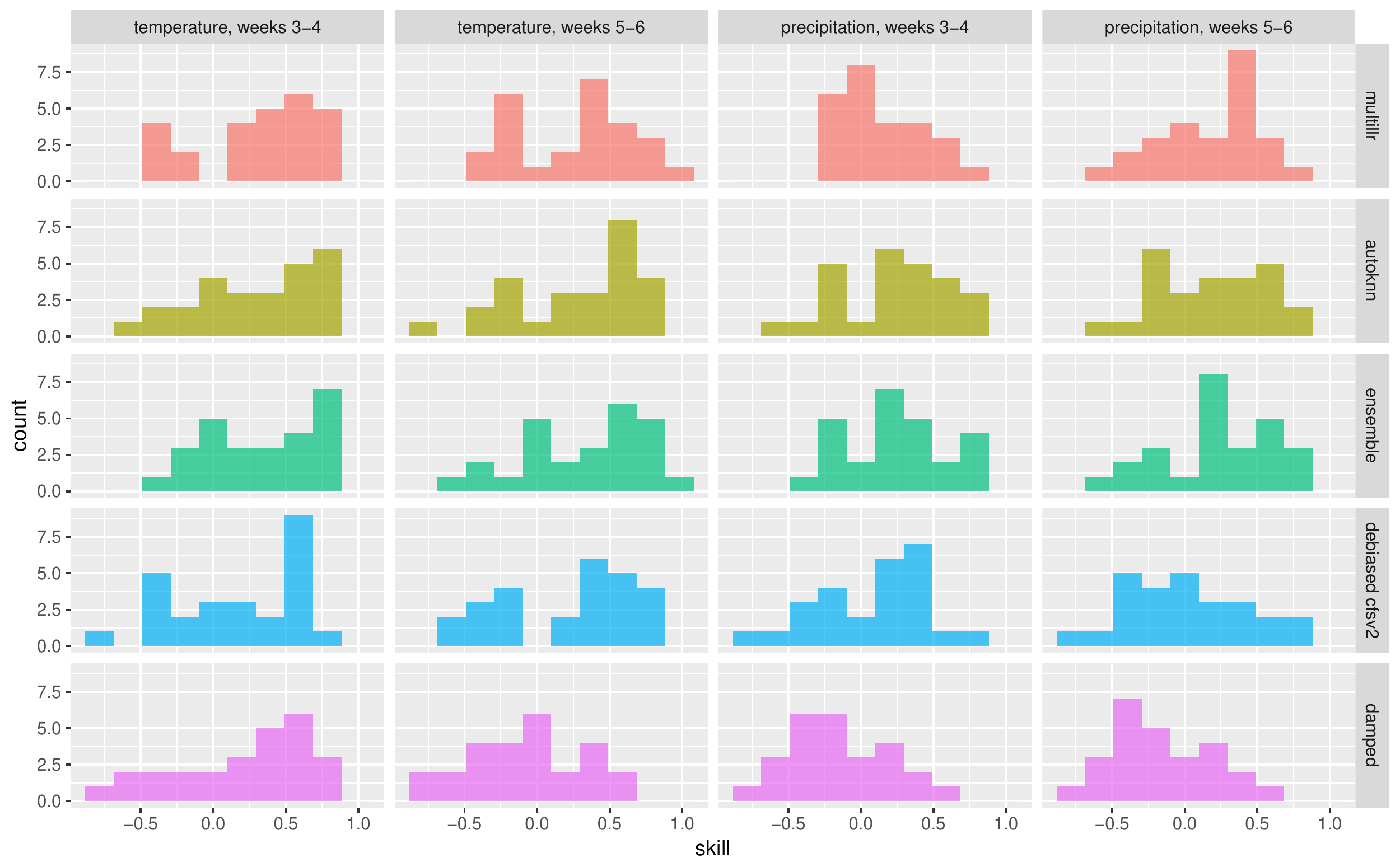}
\caption{Distribution of contest-period skills of the proposed models \stepwise and \autoknn, the proposed ensemble of \stepwise and \autoknn (\emph{ensemble}), the official contest debiased-CFSv2 baseline, and the official contest damped-persistence baseline (\emph{damped}).
See \cref{subsec:contest_period_evaluation} for more details.}
\label{fig:contest_year_histograms}
\end{figure*}
\begin{figure*}[h!]
\newcommand{\freqwidth}{0.247}%
\includegraphics[width=\freqwidth\textwidth]{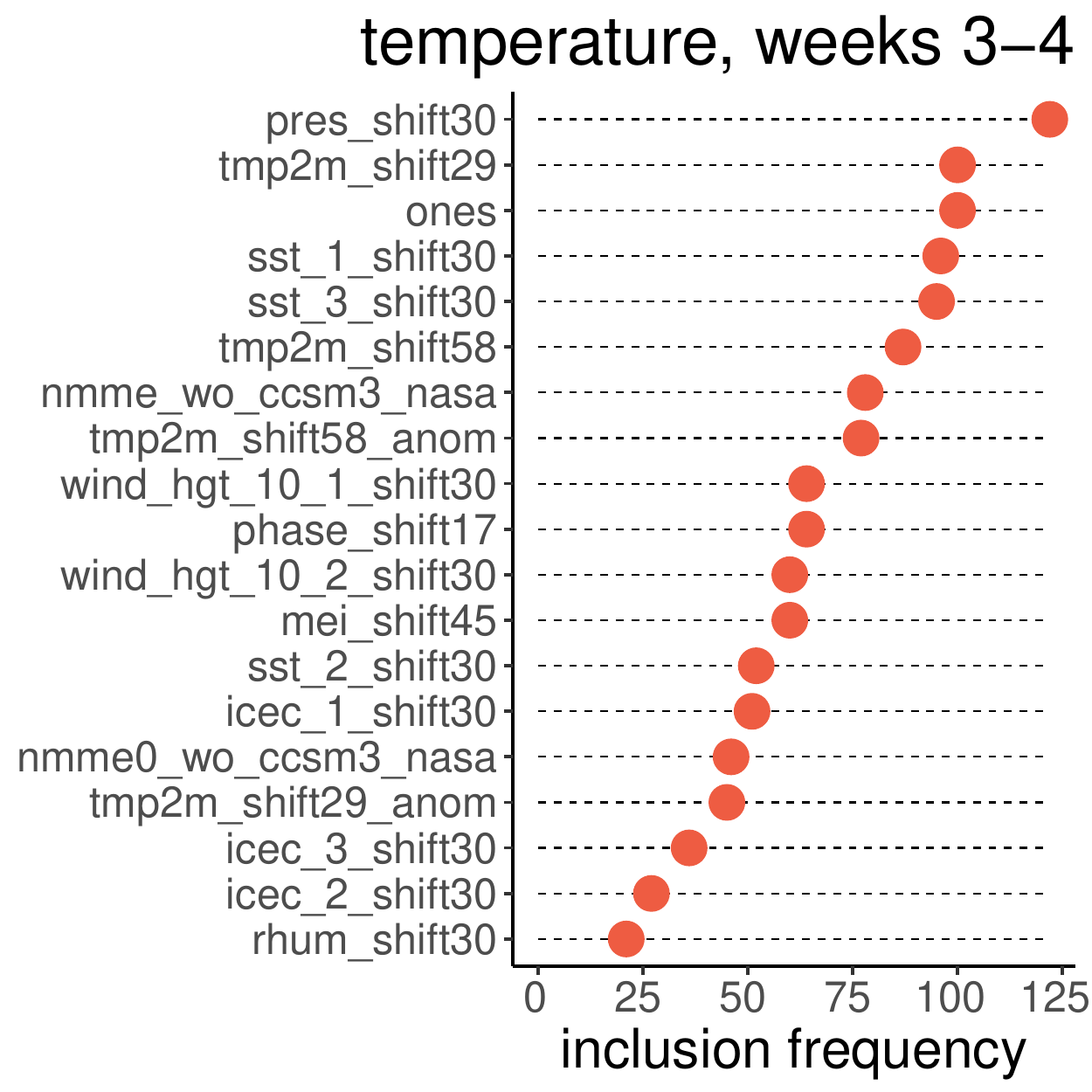}
\includegraphics[width=\freqwidth\textwidth]{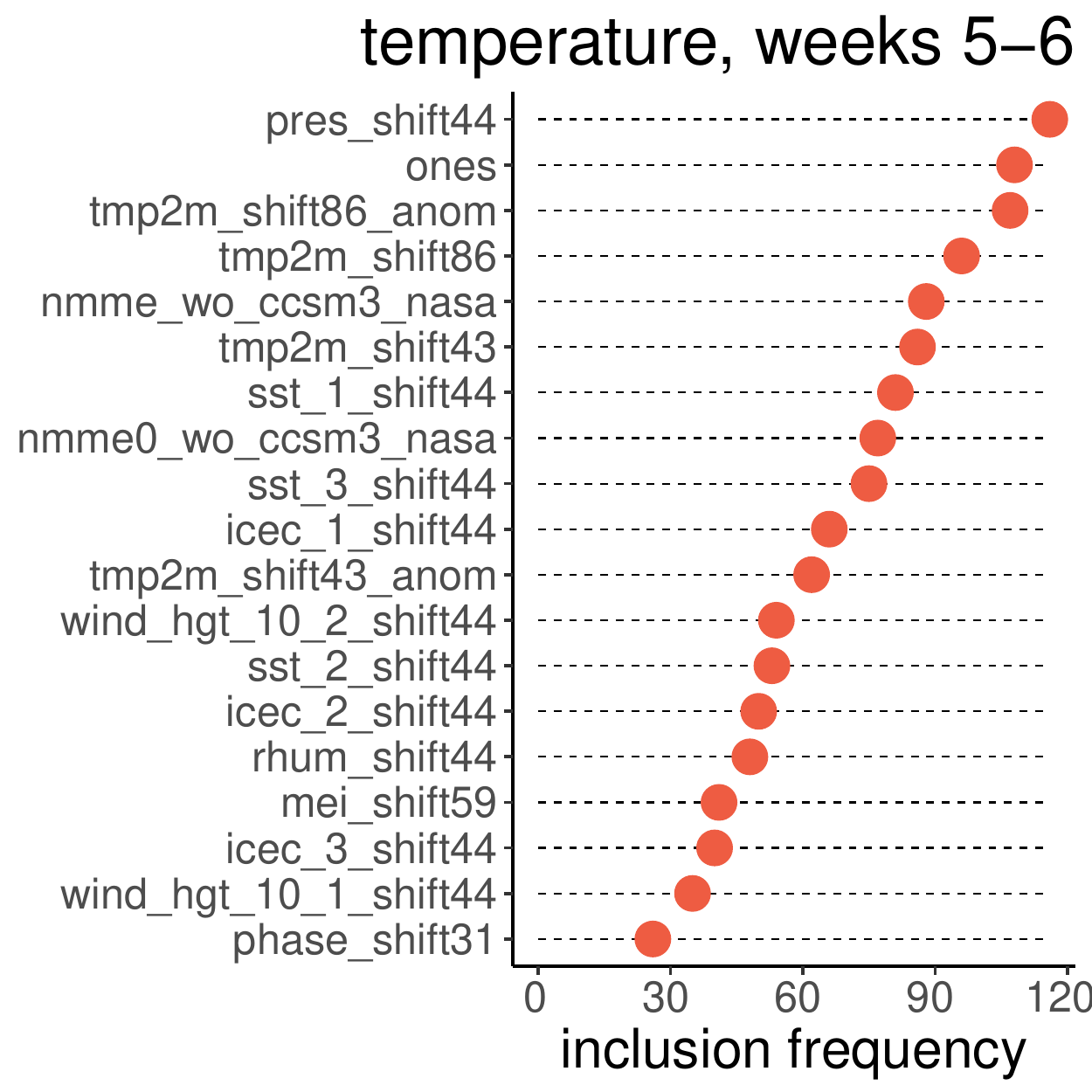}
\includegraphics[width=\freqwidth\textwidth]{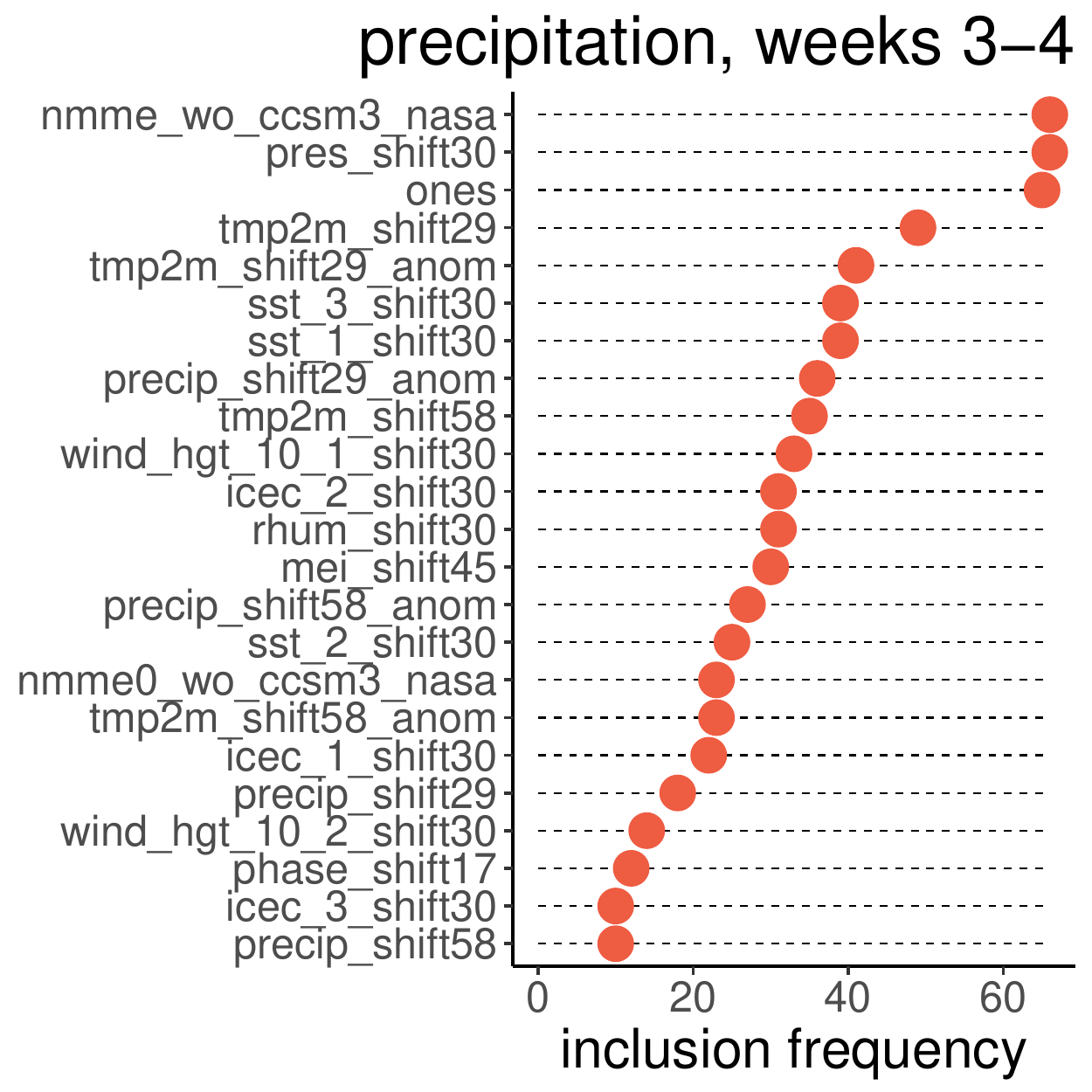}
\includegraphics[width=\freqwidth\textwidth]{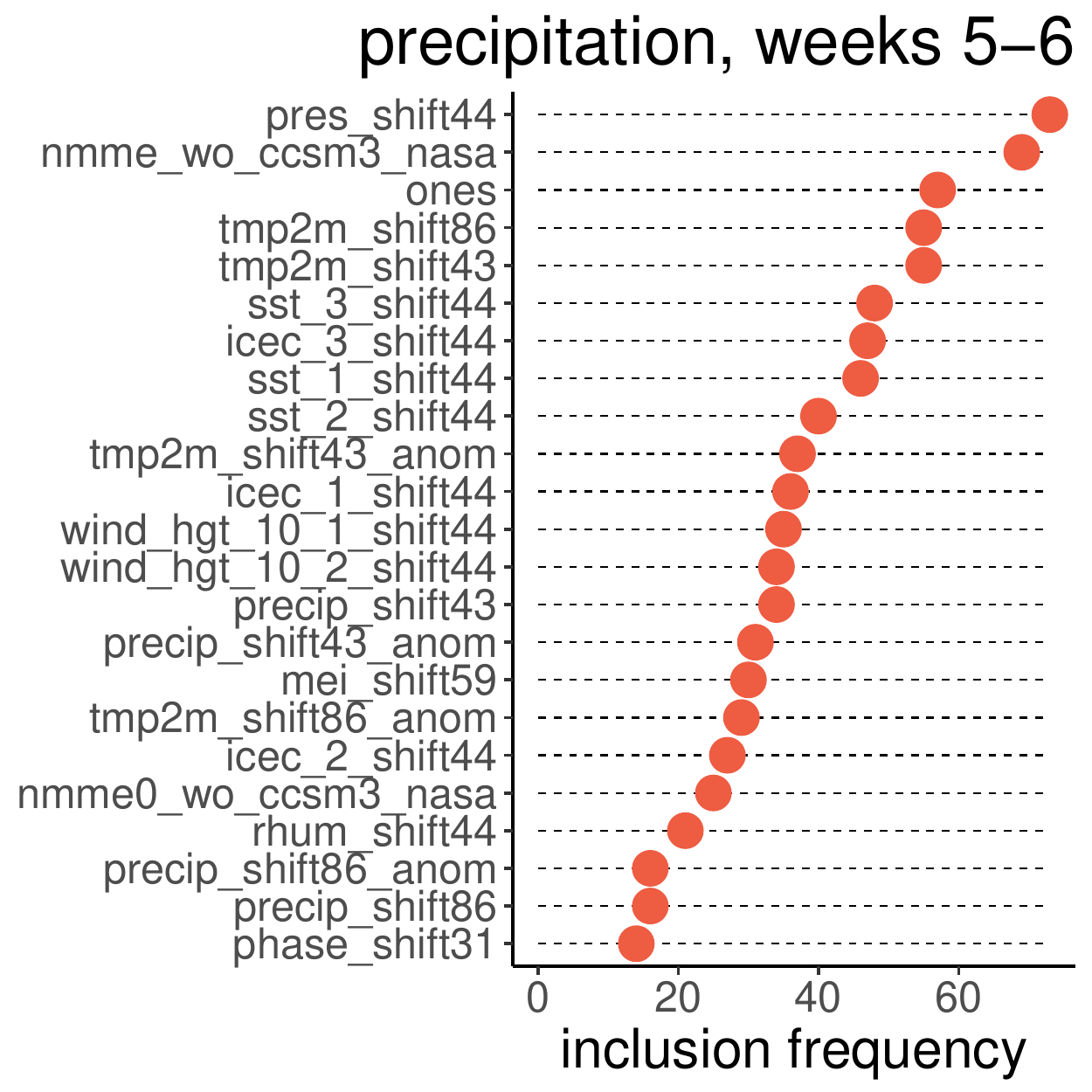}
\caption{Feature inclusion frequencies of all candidate variables for local linear regression with multitask feature selection (\stepwise) across all target dates in the historical forecast evaluation period (see \cref{sec:stepwise_explore}).}
\label{fig:stepwise_features}
\end{figure*}
\subsection{Historical Forecast Evaluation}\label{sec:historical}%
Next, we evaluate the performance of our methods over each year following the climatology period. 
That is, following the template of the contest period, we associate with each year in 2011-2017 a sequence of biweekly forecast issue dates between April 18 of that year and April 17 of the following year. For example, forecasts with submission dates between April 18, 2011 and April 17, 2012 are considered to belong to the evaluation year 2011. To mimic the actual real-time use of the forecasting system to produce forecasts for a particular target date, we train our models using only data available prior to the forecast issue date; for example, the forecasts issued on April 18, 2011 are only trained on data available prior to April 18, 2011. We compare our methods to the reconstructed debiased CFSv2 forecast.

\cref{tab:historical_skills} shows the average skills of our methods and the reconstructed debiased CFSv2 forecast (denoted by \emph{rec-deb-cfs}) in each year, 2011-2017. 
\stepwise, \autoknn, and the ensemble all achieve higher average skill than debiased CFSv2 on every task, save for \stepwise on the temperature, weeks 3-4 task.
The ensemble improves over the debiased CFSv2 average skill by 23\% for temperature weeks 3-4, by 39\% for temperature weeks 5-6, by 123\% for precipitation weeks 3-4, and by 157\% for precipitation weeks 5-6.

\cref{tab:historical_skills} also presents the average skills achieved by a three-component ensemble of \stepwise, \autoknn, and reconstructed debiased CFSv2. 
Guided by \cref{prop:ensemble}, we $\ell_2$-normalize the anomalies of each model before taking an equal-weighted average. This ensemble (denoted by \emph{ens-cfs}) produces higher average skills than the original ensemble in all prediction tasks. The \emph{ens-cfs} ensemble also substantially outperforms debiased CFSv2, with skill improvements of 40\% and 50\% for the temperature tasks and 129\% and 169\% for the precipitation tasks. These results highlight the valuable roles that ML-based models, physics-based models, and principled ensembling can all play in subseasonal forecasting.

\paragraph{Contribution of NMME}
Interestingly, the skill improvements of \autoknn were achieved without any use of physics-based model forecasts. Moreover, a \cref{prop:ensemble} ensemble of just \autoknn and \emph{rec-deb-cfsv2} realizes most of the gains of \emph{ens-cfs} without using NMME. Indeed, this ensemble has mean skills over all years in \cref{tab:historical_skills} of (temp. weeks 3-4: 0.354, temp. weeks 5-6: 0.31, precip. weeks 3-4: 0.162, precip. weeks 5-6: 0.147).

While physics-based model forecasts contribute to \stepwise through the NMME ensemble mean, \texttt{nmme\_wo\_ccsm3\_nasa} alone achieves inferior mean skill 
(temp. weeks 3-4: 0.094, temp. weeks 5-6: 0.116, precip. weeks 3-4: 0.116, precip. weeks 5-6: 0.107) over all years in \cref{tab:historical_skills} than all proposed methods and even the temperature debiased CFSv2 baseline.
One contributing factor to this performance is the mismatch between the monthly granularity of the publicly-available NMME forecasts and the biweekly granularity of our forecast periods. 
As a result, we anticipate that more granular NMME data would lead to significant improvements in the final \stepwise model.

\begin{table*}[ht]
\caption{Average skills for historical forecasts in each year following the climatology period (see \cref{sec:historical}). We compare the proposed models \stepwise and \autoknn, the proposed ensemble of \stepwise and \autoknn (\emph{ensemble}), the reconstructed debiased CFSv2 baseline (\emph{rec-deb-cfs}), and the proposed ensemble of \stepwise, \autoknn, and debiased CFSv2 (\emph{ens-cfs}).}
\label{tab:historical_skills}
\begin{tabular}{c|ccccc|ccccc}
\toprule
& \multicolumn{5}{c|}{temperature, weeks 3-4} & \multicolumn{5}{c}{temperature, weeks 5-6} \\
\midrule
year &  multillr & autoknn & ensemble & rec-deb-cfs & ens-cfs &  multillr & autoknn & ensemble & rec-deb-cfs & ens-cfs \\
\midrule
2011 & 0.2695 & 0.3664 & 0.3525 & \bf{0.4598} & 0.4589 & 0.2522 & 0.3240 & 0.3537 & 0.3879 & \bf{0.4284} \\
2012 & 0.1466 & \bf{0.3135} & 0.2548 & 0.1397 & 0.2505 & 0.2313 & \bf{0.3205} & 0.3193 & 0.1030 & 0.3033 \\
2013 & 0.1031 & 0.2011 & 0.1852 & 0.2861 & \bf{0.2878} & \bf{0.2212} & 0.0531 & 0.1833 & 0.1211 & 0.1828 \\
2014 & 0.1973 & 0.2775 & 0.2935 & 0.3018 & \bf{0.3547} & 0.1585 & 0.3056 & 0.2643 & 0.1936 & \bf{0.3297} \\
2015 & 0.3513 & 0.3885 & 0.4269 & 0.2857 & \bf{0.4404} & 0.2694 & 0.3939 & 0.3752 & 0.4234 & \bf{0.4426} \\
2016 & 0.2654 & 0.3502 & 0.3467 & 0.2490 & \bf{0.3839} & 0.2213 & 0.2882 & \bf{0.2933} & 0.0983 & 0.2720 \\
2017 & 0.3079 & 0.2807 & \bf{0.3451} & 0.0676 & 0.3253 & 0.2562 & 0.2817 & \bf{0.3025} & 0.1708 & 0.3003 \\
\hline
all & 0.2344 & 0.3111 & 0.3150 & 0.2557 & \bf{0.3573} & 0.2300 & 0.2810 & 0.2988 & 0.2142 & \bf{0.3221} \\
\bottomrule
\end{tabular}
\begin{tabular}{c|ccccc|ccccc}
\toprule
& \multicolumn{5}{c|}{precipitation, weeks 3-4} & \multicolumn{5}{c}{precipitation, weeks 5-6} \\
\midrule
year &  multillr & autoknn & ensemble & rec-deb-cfs & ens-cfs &  multillr & autoknn & ensemble & rec-deb-cfs & ens-cfs \\
\midrule
2011 & 0.1817 & 0.2173 & 0.2420 & 0.1646 & \bf{0.2692} & 0.1398 & 0.2132 & 0.2210 & 0.1835 & \bf{0.2666} \\
2012 & 0.3147 & 0.3648 & \bf{0.3983} & 0.0828 & 0.3909 & 0.3039 & 0.3943 & 0.4002 & 0.1941 & \bf{0.4224} \\
2013 & 0.1552 & 0.2026 & \bf{0.2130} & 0.0648 & 0.1711 & 0.1392 & 0.1784 & \bf{0.2031} & 0.0782 & 0.1939 \\
2014 & 0.0790 & 0.1208 & 0.1391 & 0.1272 & \bf{0.1738} & \!\!-0.0069 & \bf{0.0818} & 0.0556 & 0.0155 & 0.0782 \\
2015 & 0.0645 & \!\!-0.0053 & 0.0532 & 0.0837 & \bf{0.1043} & 0.0802 & 0.0204 & 0.0755 & 0.0292 & \bf{0.0959} \\
2016 & \bf{0.1419} & \!\!-0.0568 & 0.0636 & 0.0190 & 0.0435 & \bf{0.1703} & \!\!-0.0930 & 0.0569 & \!\!-0.0160 & 0.0483 \\
2017 & 0.1597 & 0.2156 & \bf{0.2364} & 0.0596 & 0.2250 & 0.1876 & 0.1870 & \bf{0.2315} & \!\!-0.0038 & 0.1978 \\
\hline
all & 0.1567 & 0.1513 & 0.1922 & 0.0860 & \bf{0.1968} & 0.1449 & 0.1403 & 0.1777 & 0.0691 & \bf{0.1857} \\
\bottomrule
\end{tabular}
\end{table*}

\subsection{Exploring \stepwise} \label{sec:stepwise_explore}
\cref{fig:stepwise_features} shows the frequency with which each candidate feature was selected by \stepwise in the four prediction tasks, across all target dates in the historical evaluation period. For all four tasks, the most frequently selected features include pressure (\texttt{pres}), the intercept term (\texttt{ones}), and temperature (\texttt{tmp2m}). The NMME ensemble average (\texttt{nmme\_wo\_ccsm3\_nasa}) is the first or second most commonly selected feature for predicting precipitation, but its relative selection frequency is much lower for temperature.

Although we used a slightly larger set of candidate features for the precipitation tasks---23 for precipitation, compared to 20 for temperature---the selected models are more parsimonious for precipitation than for temperature. 
The median number of selected features for predicting temperature is 7 for both forecasting horizons, while the median number of selected features for predicting precipitation is 4 for weeks 3-4 and 5 for weeks 5-6.

\subsection{Exploring \autoknn}

\cref{fig:knn_circos_precip-34w} plots the month distribution of the top nearest neighbor learned by \autoknn for predicting precipitation, weeks 3-4, as a function of the month of the target date. The figure shows that when predicting precipitation, the top neighbor for a target date is generally from the same time of year as the target date: for summer target dates, the top neighbor tends to be from a summer month and similarly for winter target dates. The corresponding plot for temperature (omitted due to space constraints) shows that this pattern does not hold when predicting temperature; rather, the top neighbors are drawn from throughout the year, regardless of the month of the target date.

The matrix plots in \cref{fig:knn_matrix_tmp2m} show the year and month of the top 20 nearest neighbors for predicting temperature, weeks 3-4, as a function of the target date. In each plot, the vertical axis ranges from $k=1$ (most similar neighbor) to $k=20$ (20th most similar neighbor). The vertical striations in both plots indicate that the top 20 neighbors for a given target date tend to be homogeneous in terms of both month and year: neighbors tend to come from the same or adjacent years and times of year. Moreover. the neighbors for post-2015 target dates tend to be from post-2010 years, in keeping with recent years' record high temperatures. The corresponding plots for precipitation (omitted due to space constraints) show that the top neighbors for precipitation do not disproportionately come from recent years, and the months of the top neighbors follow a regular seasonal pattern, consistent with \cref{fig:knn_circos_precip-34w}.

\begin{figure*}[h!]
\begin{subfigure}[b]{0.478\textwidth}
\centering
\includegraphics[width=\linewidth]{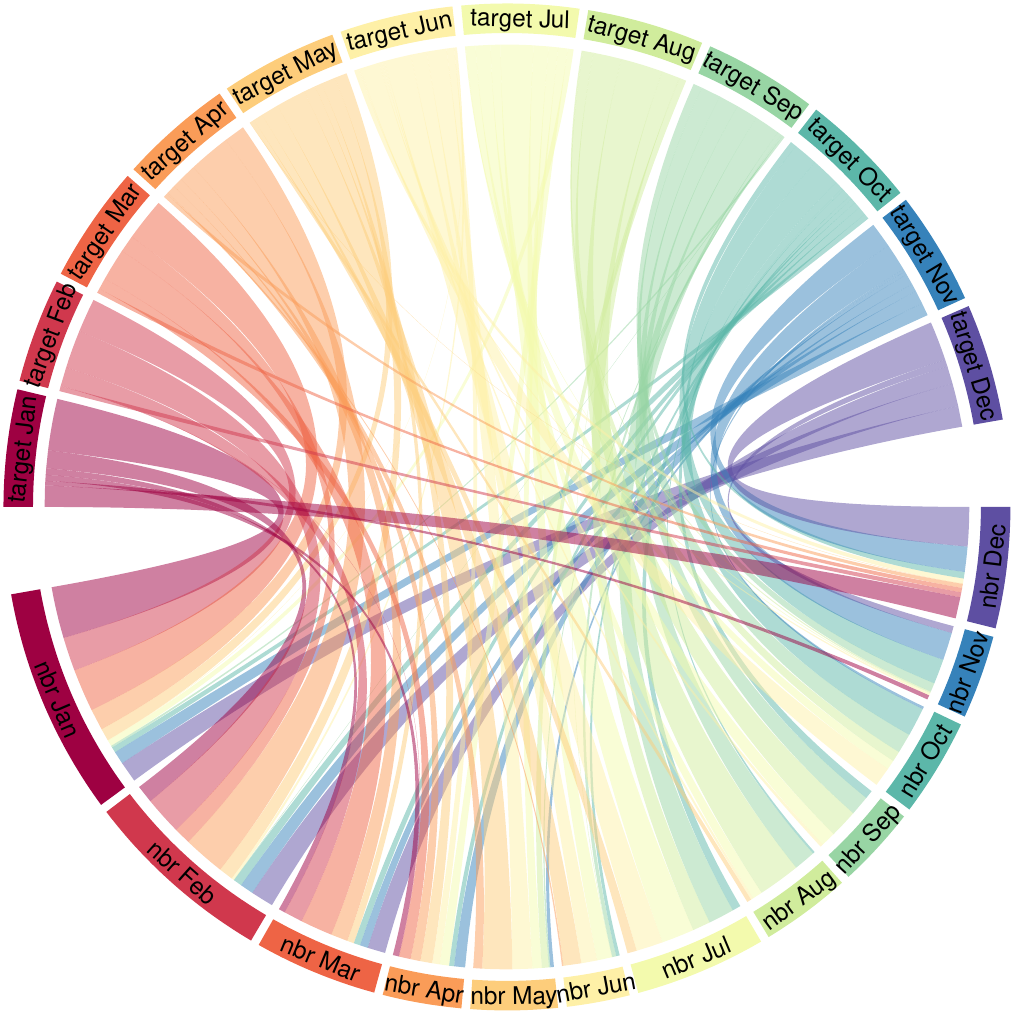}
\caption{}\label{fig:knn_circos_precip-34w}
\end{subfigure}%
~
\begin{subfigure}[b]{0.478\textwidth}
\centering
\includegraphics[width=\linewidth]{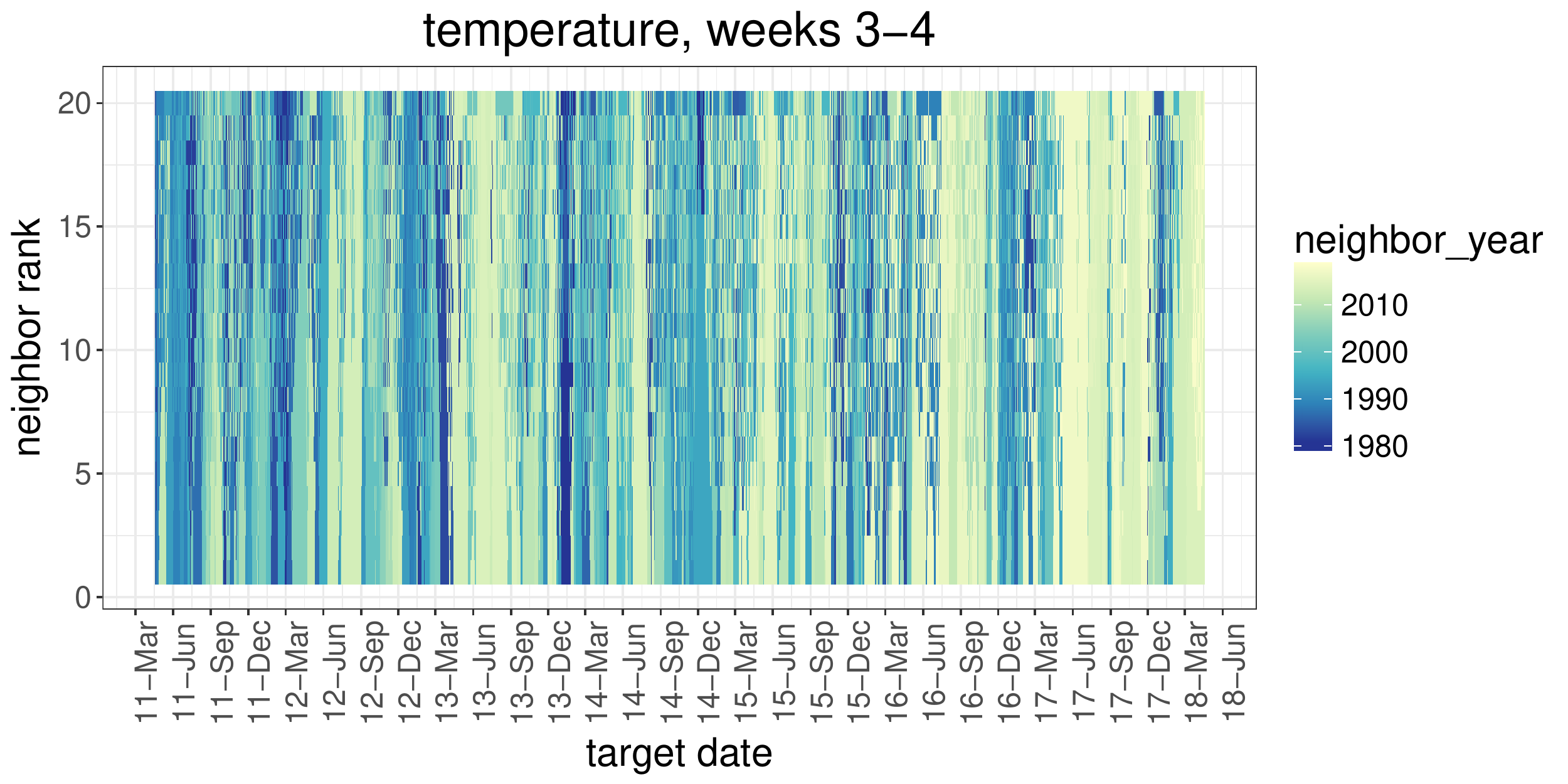}
\includegraphics[width=\linewidth]{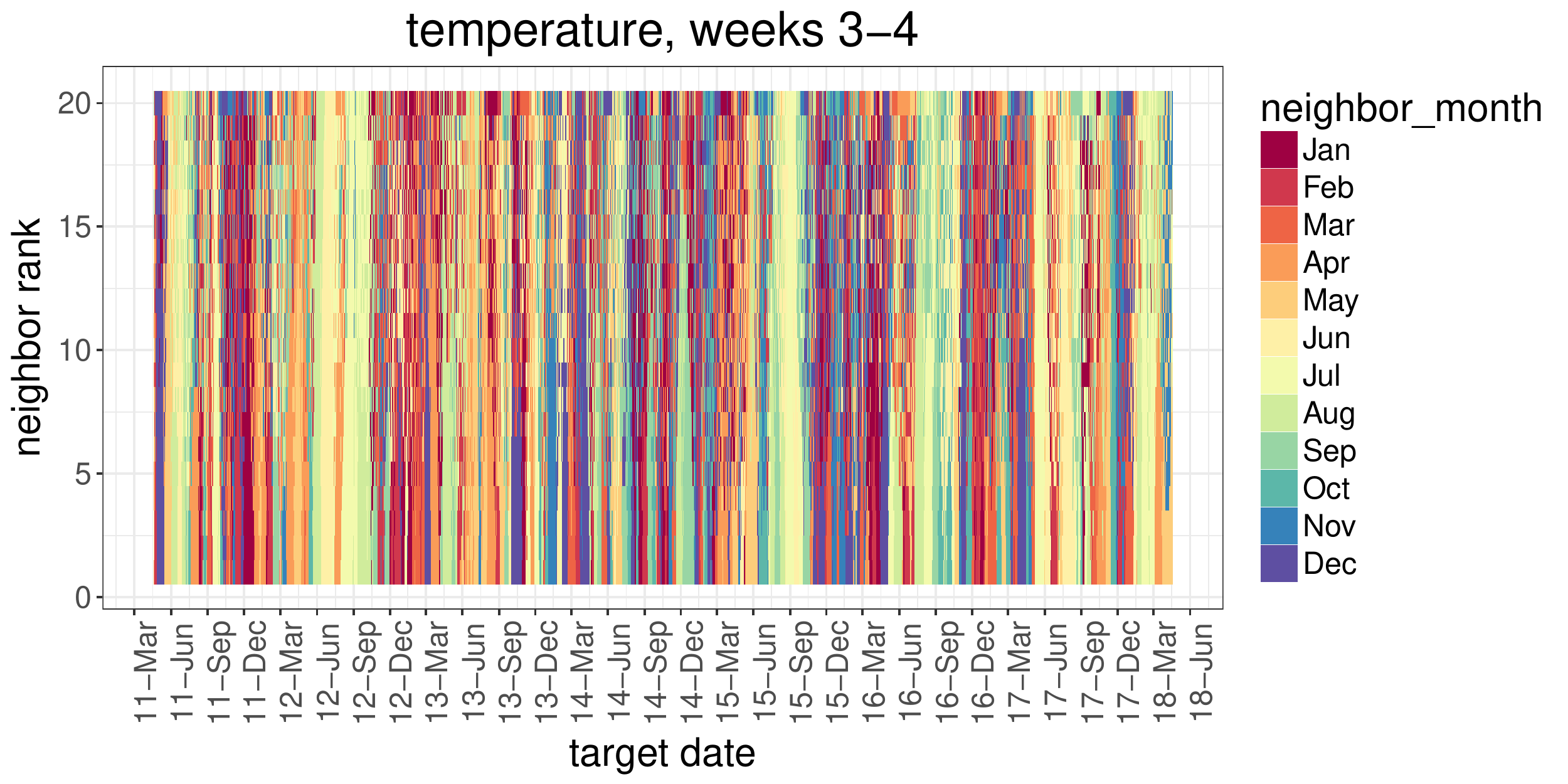}
\caption{}\label{fig:knn_matrix_tmp2m}
\end{subfigure}
\caption{(a) Precipitation, weeks 3-4: Distribution of the month of the most similar neighbor learned by \autoknn as a function of the month of the target date. (b) Temperature, weeks 3-4: Year (top) and month (bottom) of the 20 most similar neighbors learned by \autoknn (vertical axis ranges from $k=1$ to $20$) as a function of the target date (horizontal axis).}
\label{fig:knn_plots}
\end{figure*}

\section{Discussion}

To meet the USBR's Subseasonal Climate Forecast Rodeo challenge,
we developed an ML-based forecasting system and demonstrated 40-169\% improvements in forecasting skill across the challenge period (2017-18) and the years 2011-18 more generally.
Notably, the same procedures provide these improvements for each of the four Rodeo prediction tasks (forecasting temperature or precipitation at weeks 3-4 or weeks 5-6).
In the short term, we anticipate that these improvements will
benefit disaster management (e.g., anticipating droughts, floods, and other wet weather extremes) and the water management, development, and protection operations of the USBR more generally (e.g., providing irrigation water to 20\% of western U.S. farmers and generating hydroelectricity for 3.5 million homes).
In the longer term, we hope that these tools will improve our ability to anticipate and manage wildfires \citep{white2017}.

Our experience also suggests that subseasonal forecasting is fertile ground for machine learning development.
Much of the methodological novelty in our approach was driven by the unusual multitask forecasting skill objective. This objective inspired our new and provably beneficial ensembling approach and our custom multitask neighbor selection strategy.
We hope that introducing this problem to the ML community and providing a user-friendly benchmark dataset will stimulate the development and evaluation of additional subseasonal forecasting approaches.

\section*{Acknowledgments}
We thank the Subseasonal Climate Forecast Rodeo organizers for administering this challenge
and Ernest Fraenkel for bringing our team together.
JC is supported by the National Science Foundation grant AGS-1303647.

\bibliographystyle{ACM-Reference-Format}
\bibliography{refs}

\newpage\onecolumn
\appendix
\section{Supplementary \Dataset Details} \label{sec:data_supp}

The \dataset is organized as a collection of Python Pandas DataFrames and Series objects \citep{mckinney2010data} stored in HDF5 format (via \ttt{pandas.DataFrame.to\_hdf} or \ttt{pandas.Series.to\_hdf}), with one \ttt{.h5} file per DataFrame or Series.
The contents of any file can be loaded in Python using \ttt{pandas.read\_hdf}.
Each DataFrame or Series contributes data variables (features or target values) falling into one of three categories: (i) spatial (varying with the target grid point but not the target date); (ii) temporal (varying with the target date but not the target grid point); (iii) spatiotemporal (varying with both the target grid point and the target date). 
Unless otherwise noted in \cref{sec:data} or below, temporal and spatiotemporal variables arising from daily data sources were derived by averaging input values over each 14-day period, and spatial and spatiotemporal variables were derived by interpolating input data to a $1^\circ \times 1^\circ$ grid using the Climate Data Operators (CDO version 1.8.2) operator \ttt{remapdis} (distance-weighted average interpolation) with target grid \ttt{r360x181} and retaining only the contest grid points.
In addition to the variables described in \cref{sec:data}, a number of auxiliary variables were downloaded and processed but not ultimately used in our approach.
\subsection{Temperature and Precipitation Interpolation}
The downloaded temperature variables \ttt{tmin} and \ttt{tmax}, global precipitation variable \ttt{rain}, and U.S.\ precipitation variable \ttt{precip} were each interpolated to a fixed $1^\circ\times 1^\circ$ grid using the NCAR Command Language (NCL version 6.0.0) function \ttt{area\_hi2lores\_Wrap} with arguments \ttt{new\_lat  = latGlobeF(181, "lat", "latitude", "degrees\_north")}; \ttt{new\_lon = lonGlobeF(360, "lon", "longitude", "degrees\_east")}; \ttt{wgt = cos(lat*pi/180.0)} (so that points are weighted by the cosine of the latitude in radians); \ttt{opt@critpc = 50} (to require only 50\% of the values to be present to interpolate); and \ttt{fiCyclic = True} (indicating global data with longitude values that do not quite wrap around the globe). \ttt{rain} was then renamed to \ttt{precip}.

\subsection{Data Sources}
The \dataset data were downloaded from the following sources.
\bitemize
\item Temperature \citep{fan2008}: \url{ftp://ftp.cpc.ncep.noaa.gov/precip/PEOPLE/wd52ws/global_temp/}
\item Global precipitation \citep{xie2010cpc}: \url{ftp://ftp.cpc.ncep.noaa.gov/precip/CPC_UNI_PRCP/GAUGE_GLB/RT/}
\item \US precipitation \citep{xie2010cpc}: \url{https://www.esrl.noaa.gov/psd/thredds/catalog/Datasets/cpc_us_precip/catalog.html}
\item Sea surface temperature and sea ice concentration \citep{reynolds2007daily}: \url{ftp://ftp.cdc.noaa.gov/Projects/Datasets/noaa.oisst.v2.highres/}
\item Multivariate ENSO index (MEI) \citep{wolter1993monitoring,wolter1998measuring,zimmerman2016nip}: \url{https://www.esrl.noaa.gov/psd/enso/mei/}
\item Madden-Julian oscillation (MJO) \citep{wheeler2004all}: \url{http://www.bom.gov.au/climate/mjo/graphics/rmm.74toRealtime.txt}
\item Relative humidity, sea level pressure, and precipitable water for entire atmosphere \citep{kalnay1996ncep}: \url{ftp://ftp.cdc.noaa.gov/Datasets/ncep.reanalysis/surface/}
\item Pressure and potential evaporation \citep{kalnay1996ncep}: \url{ftp://ftp.cdc.noaa.gov/Datasets/ncep.reanalysis/surface_gauss/}
\item Geopotential height, zonal wind, and longitudinal wind \citep{kalnay1996ncep}: \url{ftp://ftp.cdc.noaa.gov/Datasets/ncep.reanalysis.dailyavgs/pressure/}
\item North American Multi-Model Ensemble (NMME) \citep{kirtman2014north}: \url{https://iridl.ldeo.columbia.edu/SOURCES/.Models/.NMME/}
\item Elevation \citep{gates1975new}: \url{http://research.jisao.washington.edu/data_sets/elevation/elev.1-deg.nc}
\item K\"oppen-Geiger climate classifications \citep{kottek2006world}: \url{http://koeppen-geiger.vu-wien.ac.at/present.htm}
 \eitemize

\subsection{Dataset Files}
Below, we list the contents of each \dataset file.
Each file with the designation `Series' contains a Pandas Series object with a MultiIndex for the target latitude (\ttt{lat}), longitude (\ttt{lon}), and date defining the start of the target two-week period (\ttt{start\_date}).
Each file with the designation `MultiIndex DataFrame' contains a Pandas DataFrame object with a MultiIndex for \ttt{lat}, \ttt{lon}, and \ttt{start\_date}.
Each file with a filename beginning with `nmme' contains a Pandas DataFrame object with \ttt{target\_start}, \ttt{lat}, and \ttt{lon} columns; the \ttt{target\_start} column plays the same role as \ttt{start\_date} in other files, indicating the date defining the start of the target two-week period.
Each remaining file with the designation `DataFrame' contains a Pandas DataFrame object with \ttt{lat} and \ttt{lon} columns if the contained variables are spatial; a \ttt{start\_date} column if the contained variables are temporal; and \ttt{start\_date}, \ttt{lat}, and \ttt{lon} columns if the contained variables are spatiotemporal.

The filename prefix `gt-wide' indicates that a file contains temporal variables representing a base variable's measurement at multiple locations on a latitude-longitude grid that need not correspond to contest grid point locations. The temporal variable column names are tuples in the format (`\emph{base variable name}', \emph{latitude}, \emph{longitude}). 
The base variable measurements underlying the files with the filename prefix `gt-wide\_contest' were first interpolated to a $1^\circ\times 1^\circ$ grid. The measurements underlying the remaining `gt-wide' files did not undergo interpolation; the original data source grids were instead employed.
\begin{multicols}{2}
{\scriptsize
\bitemize[leftmargin=*]
\item \ttt{gt-climate\_regions.h5} (DataFrame)
\bitemize
\item Spatial variable K\"oppen-Geiger climate classifications (\ttt{climate\_region})
\eitemize
\item \ttt{gt-contest\_pevpr.sfc.gauss-14d-1948-2018.h5} (Series)%
\bitemize
\item Spatiotemporal variable potential evaporation (\ttt{pevpr.sfc.gauss})
\eitemize
\item \ttt{gt-contest\_precip-14d-1948-2018.h5} (Series)%
\bitemize
\item Spatiotemporal variable precipitation (\ttt{precip}) 
\eitemize
\item \ttt{gt-contest\_pres.sfc.gauss-14d-1948-2018.h5} (Series)%
\bitemize
\item Spatiotemporal variable pressure (\ttt{pres.sfc.gauss})
\eitemize
\item \ttt{gt-contest\_pr\_wtr.eatm-14d-1948-2018.h5} (Series)%
\bitemize
\item Spatiotemporal variable precipitable water for entire atmosphere (\ttt{pr\_wtr.eatm})
\eitemize
\item \ttt{gt-contest\_rhum.sig995-14d-1948-2018.h5} (Series)%
\bitemize
\item Spatiotemporal variable relative humidity (\ttt{rhum.sig995})
\eitemize
\item \ttt{gt-contest\_slp-14d-1948-2018.h5} (Series)%
\bitemize
\item Spatiotemporal variable sea level pressure (\ttt{slp})
\eitemize
\item \ttt{gt-contest\_tmax-14d-1979-2018.h5} (Series)%
\bitemize
\item Spatiotemporal variable maximum temperature at 2m (\ttt{tmax})
\eitemize
\item \ttt{gt-contest\_tmin-14d-1979-2018.h5} (Series)%
\bitemize
\item Spatiotemporal variable minimum temperature at 2m (\ttt{tmin})
\eitemize
\item \ttt{gt-contest\_tmp2m-14d-1979-2018.h5} (DataFrame)%
\bitemize
\item Spatiotemporal variables temperature at 2m (\ttt{tmp2m}), average squared temperature at 2m over two-week period (\ttt{tmp2m\_sqd}), and standard deviation of temperature at 2m over two-week period (\ttt{tmp2m\_std})
\eitemize
\item \ttt{gt-contest\_wind\_hgt\_100-14d-1948-2018.h5} (Series)%
\bitemize
\item Spatiotemporal variable geopotential height at 100 millibars (\texttt{contest\_wind\_hgt\_100}) 
\eitemize
\item \ttt{gt-contest\_wind\_hgt\_10-14d-1948-2018.h5} (Series)%
\bitemize
\item Spatiotemporal variable geopotential height at 10 millibars (\texttt{contest\_wind\_hgt\_10}) 
\eitemize
\item \ttt{gt-contest\_wind\_hgt\_500-14d-1948-2018.h5} (Series)%
\bitemize
\item Spatiotemporal variable geopotential height at 500 millibars (\texttt{contest\_wind\_hgt\_500}) 
\eitemize
\item \ttt{gt-contest\_wind\_hgt\_850-14d-1948-2018.h5} (Series)%
\bitemize
\item Spatiotemporal variable geopotential height at 850 millibars (\texttt{contest\_wind\_hgt\_850}) 
\eitemize
\item \ttt{gt-contest\_wind\_uwnd\_250-14d-1948-2018.h5} (Series)%
\bitemize
\item Spatiotemporal variable zonal wind at 250 millibars (\texttt{contest\_wind\_uwnd\_250}) 
\eitemize
\item \ttt{gt-contest\_wind\_uwnd\_925-14d-1948-2018.h5} (Series)%
\bitemize
\item Spatiotemporal variable zonal wind at 925 millibars (\texttt{contest\_wind\_uwnd\_925})
\eitemize
\item \ttt{gt-contest\_wind\_vwnd\_250-14d-1948-2018.h5} (Series)%
\bitemize
\item Spatiotemporal variable longitudinal wind at 250 millibars (\texttt{contest\_wind\_vwnd\_250})
\eitemize
\item \ttt{gt-contest\_wind\_vwnd\_925-14d-1948-2018.h5} (Series)%
\bitemize
\item Spatiotemporal variable longitudinal wind at 925 millibars (\texttt{contest\_wind\_vwnd\_925}) 
\eitemize
\item \ttt{gt-elevation.h5} (DataFrame)
\bitemize
\item Spatial variable elevation (\ttt{elevation})
\eitemize
\item \ttt{gt-mei-1950-2018.h5} (DataFrame)%
\bitemize
\item Temporal variables MEI (\ttt{mei}), MEI rank (\ttt{rank}), and Ni{\~n}o Index Phase (\ttt{nip}) derived from \ttt{mei} and \ttt{rank} using the definition in \citep{zimmerman2016nip}
\eitemize
\item \ttt{gt-mjo-1d-1974-2018.h5} (DataFrame)%
\bitemize
\item Temporal variables MJO phase (\ttt{phase}) and amplitude (\ttt{amplitude})
\eitemize
\item \ttt{gt-pca\_icec\_2010-14d-1981-2018.h5} (DataFrame)%
\bitemize
\item Temporal variables top PCs of \ttt{gt-wide\_contest\_icec-14d-1981-2018.h5} based on PC loadings from 1981-2010 %
\eitemize
\item \ttt{gt-pca\_sst\_2010-14d-1981-2018.h5} (DataFrame)%
\bitemize
\item Temporal variables top PCs of \ttt{gt-wide\_contest\_sst-14d-1981-2018.h5} based on PC loadings from 1981-2010 %
\eitemize
\item \ttt{gt-pca\_wind\_hgt\_100\_2010-14d-1948-2018.h5} (DataFrame)%
\bitemize
\item Temporal variables top PCs of \ttt{gt-wide\_wind\_hgt\_100-14d-1948-2018.h5} based on PC loadings from 1948-2010 %
\eitemize
\item \ttt{gt-pca\_wind\_hgt\_10\_2010-14d-1948-2018.h5} (DataFrame)%
\bitemize
\item Temporal variables top PCs of \ttt{gt-wide\_wind\_hgt\_10-14d-1948-2018.h5} based on PC loadings from 1948-2010 %
\eitemize
\item \ttt{gt-pca\_wind\_hgt\_500\_2010-14d-1948-2018.h5} (DataFrame)%
\bitemize
\item Temporal variables top PCs of \ttt{gt-wide\_wind\_hgt\_500-14d-1948-2018.h5} based on PC loadings from 1948-2010 %
\eitemize
\item \ttt{gt-pca\_wind\_hgt\_850\_2010-14d-1948-2018.h5} (DataFrame)%
\bitemize
\item Temporal variables top PCs of \ttt{gt-wide\_wind\_hgt\_850-14d-1948-2018.h5} based on PC loadings from 1948-2010 %
\eitemize
\item \ttt{gt-pca\_wind\_uwnd\_250\_2010-14d-1948-2018.h5} (DataFrame)%
\bitemize
\item Temporal variables top PCs of \ttt{gt-wide\_wind\_uwnd\_250-14d-1948-2018.h5} based on PC loadings from 1948-2010 %
\eitemize
\item \ttt{gt-pca\_wind\_uwnd\_925\_2010-14d-1948-2018.h5} (DataFrame)%
\bitemize
\item Temporal variables top PCs of \ttt{gt-wide\_wind\_uwnd\_925-14d-1948-2018.h5} based on PC loadings from 1948-2010 %
\eitemize
\item \ttt{gt-pca\_wind\_vwnd\_250\_2010-14d-1948-2018.h5} (DataFrame)%
\bitemize
\item Temporal variables top PCs of \ttt{gt-wide\_wind\_vwnd\_250-14d-1948-2018.h5} based on PC loadings from 1948-2010 %
\eitemize
\item \ttt{gt-pca\_wind\_vwnd\_925\_2010-14d-1948-2018.h5} (DataFrame)%
\bitemize
\item Temporal variables top PCs of \ttt{gt-wide\_wind\_vwnd\_925-14d-1948-2018.h5} based on PC loadings from 1948-2010 %
\eitemize
\item \ttt{gt-wide\_contest\_icec-14d-1981-2018.h5} (DataFrame)%
\bitemize
\item Temporal variables sea ice concentration for all grid points in the Pacific basin (20S to 65N, 150E to 90W) ((`{icec}',\emph{latitude},\emph{longitude}))
\eitemize
\item \ttt{gt-wide\_contest\_sst-14d-1981-2018.h5} (DataFrame)%
\bitemize
\item Temporal variables sea surface temperature for all grid points in the Pacific basin (20S to 65N, 150E to 90W) ((`{sst}',\emph{latitude},\emph{longitude})) 
\eitemize
\item \ttt{gt-wide\_wind\_hgt\_100-14d-1948-2018.h5} (DataFrame)%
\bitemize
\item Temporal variables geopotential height at 100 millibars for all grid points globally ((`{wind\_hgt\_100}',\emph{latitude},\emph{longitude})) 
\eitemize
\item \ttt{gt-wide\_wind\_hgt\_10-14d-1948-2018.h5} (DataFrame)%
\bitemize
\item Temporal variables geopotential height at 10 millibars for all grid points globally ((`{wind\_hgt\_10}',\emph{latitude},\emph{longitude})) 
\eitemize
\item \ttt{gt-wide\_wind\_hgt\_500-14d-1948-2018.h5} (DataFrame)%
\bitemize
\item Temporal variables geopotential height at 500 millibars for all grid points globally ((`{wind\_hgt\_500}',\emph{latitude},\emph{longitude})) 
\eitemize
\item \ttt{gt-wide\_wind\_hgt\_850-14d-1948-2018.h5} (DataFrame)%
\bitemize
\item Temporal variables geopotential height at 850 millibars for all grid points globally ((`{wind\_hgt\_850}',\emph{latitude},\emph{longitude})) 
\eitemize
\item \ttt{gt-wide\_wind\_uwnd\_250-14d-1948-2018.h5} (DataFrame)%
\bitemize
\item Temporal variables zonal wind at 250 millibars for all grid points globally ((`{wind\_uwnd\_250}',\emph{latitude},\emph{longitude})) 
\eitemize
\item \ttt{gt-wide\_wind\_uwnd\_925-14d-1948-2018.h5} (DataFrame)%
\bitemize
\item Temporal variables zonal wind at 925 millibars for all grid points globally ((`{wind\_uwnd\_925}',\emph{latitude},\emph{longitude})) 
\eitemize
\item \ttt{gt-wide\_wind\_vwnd\_250-14d-1948-2018.h5} (DataFrame)%
\bitemize
\item Temporal variables longitudinal wind at 250 millibars for all grid points globally ((`{wind\_vwnd\_250}',\emph{latitude},\emph{longitude})) 
\eitemize
\item \ttt{gt-wide\_wind\_vwnd\_925-14d-1948-2018.h5} (DataFrame)%
\bitemize
\item Temporal variables longitudinal wind at 925 millibars for all grid points globally ((`{wind\_vwnd\_925}',\emph{latitude},\emph{longitude})) 
\eitemize
\item \ttt{nmme0-prate-34w-1982-2018.h5} (DataFrame)%
\bitemize
\item Spatiotemporal variables most recent monthly NMME model forecasts for \ttt{precip} (\ttt{cancm3\_0}, \ttt{cancm4\_0}, \ttt{ccsm3\_0}, \ttt{ccsm4\_0}, \ttt{cfsv2\_0}, \ttt{gfdl-flor-a\_0}, \ttt{gfdl-flor-b\_0}, \ttt{gfdl\_0}, '\ttt{nasa\_0}, '\ttt{nmme0\_mean}) and average forecast across those models (\ttt{nmme0\_mean})
\eitemize
\item \ttt{nmme0-prate-56w-1982-2018.h5} (DataFrame)%
\bitemize
\item Spatiotemporal variables most recent monthly NMME model forecasts for \ttt{precip} (\ttt{cancm3\_0}, \ttt{cancm4\_0}, \ttt{ccsm3\_0}, \ttt{ccsm4\_0}, \ttt{cfsv2\_0}, \ttt{gfdl-flor-a\_0}, \ttt{gfdl-flor-b\_0}, \ttt{gfdl\_0}, '\ttt{nasa\_0}, '\ttt{nmme0\_mean}) and average forecast across those models (\ttt{nmme0\_mean})
\eitemize
\item \ttt{nmme0-tmp2m-34w-1982-2018.h5} (DataFrame)%
\bitemize
\item Spatiotemporal variables most recent monthly NMME model forecasts for \ttt{tmp2m} (\ttt{cancm3\_0}, \ttt{cancm4\_0}, \ttt{ccsm3\_0}, \ttt{ccsm4\_0}, \ttt{cfsv2\_0}, \ttt{gfdl-flor-a\_0}, \ttt{gfdl-flor-b\_0}, \ttt{gfdl\_0}, '\ttt{nasa\_0}, '\ttt{nmme0\_mean}) and average forecast across those models (\ttt{nmme0\_mean})
\eitemize
\item \ttt{nmme0-tmp2m-56w-1982-2018.h5} (DataFrame)%
\bitemize
\item Spatiotemporal variables most recent monthly NMME model forecasts for \ttt{tmp2m} (\ttt{cancm3\_0}, \ttt{cancm4\_0}, \ttt{ccsm3\_0}, \ttt{ccsm4\_0}, \ttt{cfsv2\_0}, \ttt{gfdl-flor-a\_0}, \ttt{gfdl-flor-b\_0}, \ttt{gfdl\_0}, '\ttt{nasa\_0}, '\ttt{nmme0\_mean}) and average forecast across those models (\ttt{nmme0\_mean})
\eitemize
\item \ttt{nmme-prate-34w-1982-2018.h5} (DataFrame)%
\bitemize
\item Spatiotemporal variables weeks 3-4 weighted average of monthly NMME model forecasts for \ttt{precip} (\ttt{cancm3}, \ttt{cancm4}, \ttt{ccsm3}, \ttt{ccsm4}, \ttt{cfsv2}, \ttt{gfdl}, \ttt{gfdl-flor-a}, \ttt{gfdl-flor-b}, \ttt{nasa}) and average forecast across those models (\ttt{nmme\_mean})
\eitemize
\item \ttt{nmme-prate-56w-1982-2018.h5} (DataFrame)%
\bitemize
\item Spatiotemporal variables weeks 5-6 weighted average of monthly NMME model forecasts for \ttt{precip} (\ttt{cancm3}, \ttt{cancm4}, \ttt{ccsm3}, \ttt{ccsm4}, \ttt{cfsv2}, \ttt{gfdl}, \ttt{gfdl-flor-a}, \ttt{gfdl-flor-b}, \ttt{nasa}) and average forecast across those models (\ttt{nmme\_mean})
\eitemize
\item \ttt{nmme-tmp2m-34w-1982-2018.h5} (DataFrame)%
\bitemize
\item Spatiotemporal variables weeks 3-4 weighted average of monthly NMME model forecasts for \ttt{tmp2m} (\ttt{cancm3}, \ttt{cancm4}, \ttt{ccsm3}, \ttt{ccsm4}, \ttt{cfsv2}, \ttt{gfdl}, \ttt{gfdl-flor-a}, \ttt{gfdl-flor-b}, \ttt{nasa}) and average forecast across those models (\ttt{nmme\_mean})
\eitemize
\item \ttt{nmme-tmp2m-56w-1982-2018.h5} (DataFrame)%
\bitemize
\item Spatiotemporal variables weeks 5-6 weighted average of monthly NMME model forecasts for \ttt{tmp2m} (\ttt{cancm3}, \ttt{cancm4}, \ttt{ccsm3}, \ttt{ccsm4}, \ttt{cfsv2}, \ttt{gfdl}, \ttt{gfdl-flor-a}, \ttt{gfdl-flor-b}, \ttt{nasa}) and average forecast across those models (\ttt{nmme\_mean})
\eitemize
\item \ttt{official\_climatology-contest\_precip-1981-2010.h5} (DataFrame)%
\bitemize
\item Spatiotemporal variable precipitation climatology (\ttt{precip\_clim}). Only the dates 1799-12-19--1800-12-18 are included as representatives of each (non-leap day) month-day combination.
\eitemize
\item \ttt{official\_climatology-contest\_tmp2m-1981-2010.h5} (DataFrame)%
\bitemize
\item Spatiotemporal variable temperature at 2 meters climatology (\ttt{tmp2m\_clim}). Only the dates 1799-12-19--1800-12-18 are included as representatives of each (non-leap day) month-day combination. 
\eitemize
\eitemize
}
\end{multicols}
\section{Debiased CFSv2 Reconstruction Details}
\label{sec:cfs_supp}

For the target dates in the 2011-2018 historical forecast evaluation period of \cref{sec:historical}, Climate Forecast System (CFSv2) archived operational forecasts were retrieved from the National Center for Environmental Information (NCEI) site at \url{https://nomads.ncdc.noaa.gov/modeldata/cfsv2_forecast_ts_9mon/}. The Gaussian gridded data (approximately 0.93$^\circ$ resolution) for precipitation rate and 2-meter temperature were interpolated to the Rodeo forecast grid at 1$^\circ$ resolution. These data were then extracted as a window from 25 to 50 N and -125 to -93 W.
Data were extracted for all forecast issue dates, for each cardinal hour (00, 06, 12, and 18 UTC). Interpolation from Gaussian grid to regular $1^\circ \times 1^\circ$ latitude longitude grids was accomplished using a bilinear interpolation under the Python Basemap package \citep{hunter2007matplotlib}. Any missing data were replaced by the average measurement from the available forecasts in the 2-week period.

To obtain a suitable long-term average for debiasing our reconstructed CFSv2 forecasts, precipitation (prate\_f) and temperature (tmp2m\_f) CFS Reforecast data from 1999-2010 were downloaded from \url{https://nomads.ncdc.noaa.gov/data/cfsr-hpr-ts45/}; interpolated to a $1^\circ \times 1^\circ$ grid via bilinear interpolation using wgrib2 (v0.2.0.6c) with arguments \ttt{-new\_grid\_winds earth} and \ttt{-new\_grid ncep grid 3}; and then restricted to the contest region. Temperatures were converted from Kelvin to Celsius, and the precipitation measurements were scaled from mm/s to mm/2-week period. Finally, each 2-week period in the data was averaged (for temperature) or summed (for precipitation). Any missing data were replaced by the average measurement from the available forecasts in the 2-week period.

\clearpage
\opt{arxiv}{ \section{Supplementary Figures}
\label{sec:supplementary_figures}
\begin{figure*}[h!]
\includegraphics[width=0.45\textwidth]{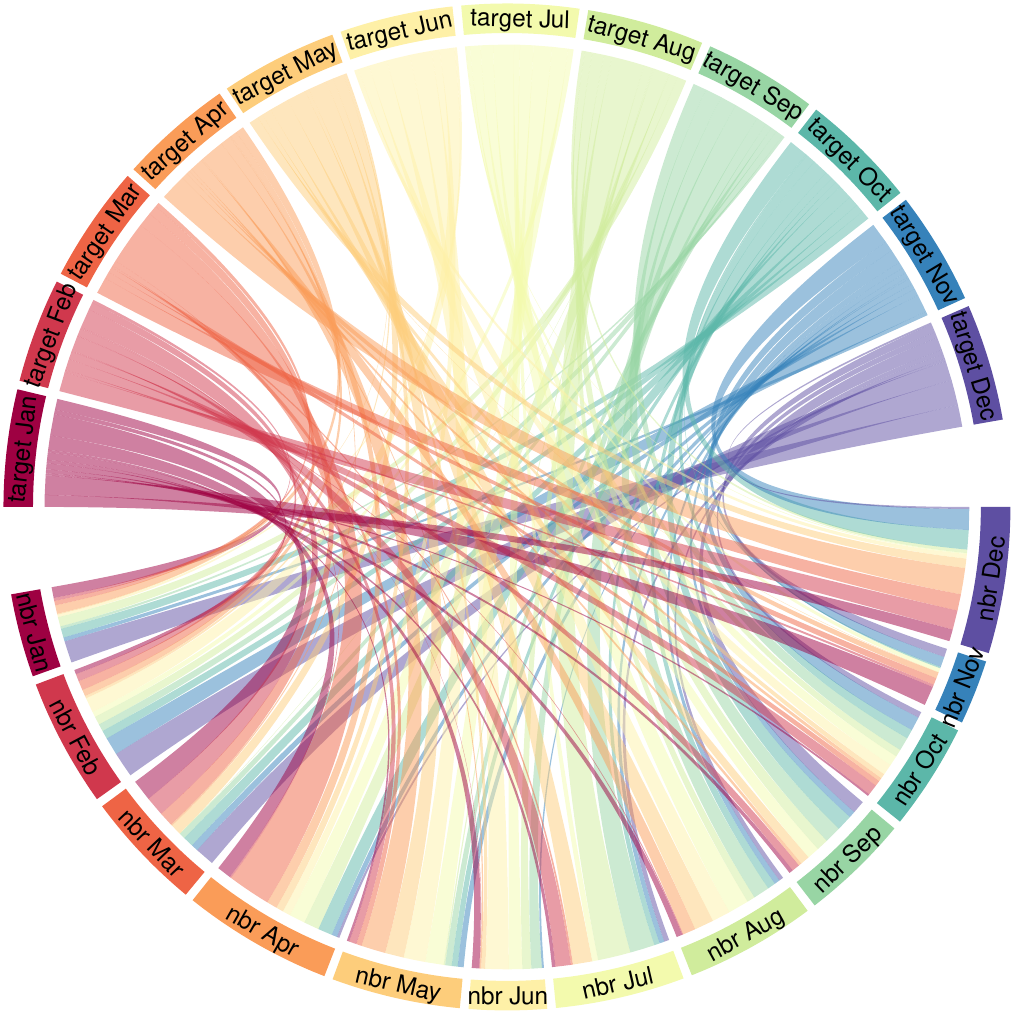}
\includegraphics[width=0.45\textwidth]{knn_circos_plots/circos_precip-34w_maxnbrs1_all_dates-2011-2017.pdf}
\caption{Distribution of the month of the most similar neighbor learned by \autoknn as a function of the month of the target date to be predicted. Left: Most similar neighbor for temperature, weeks 3-4. Right: Most similar neighbor for precipitation, weeks 3-4. The plots for weeks 5-6 are very similar. For temperature, the most similar neighbor can come from any time of year, regardless of the month of the target date. For precipitation, we instead observe a strong seasonal pattern; the season of the most similar neighbor generally matches the season of the target date.}
\label{fig:knn_circos_plot}
\end{figure*}

\begin{figure*}[h!]
\centering
\includegraphics[width=\textwidth]{knn_matrix_plots/matplot_tmp2m-34w_maxnbrs20_all_dates-2011-2017.pdf}
\includegraphics[width=\textwidth]{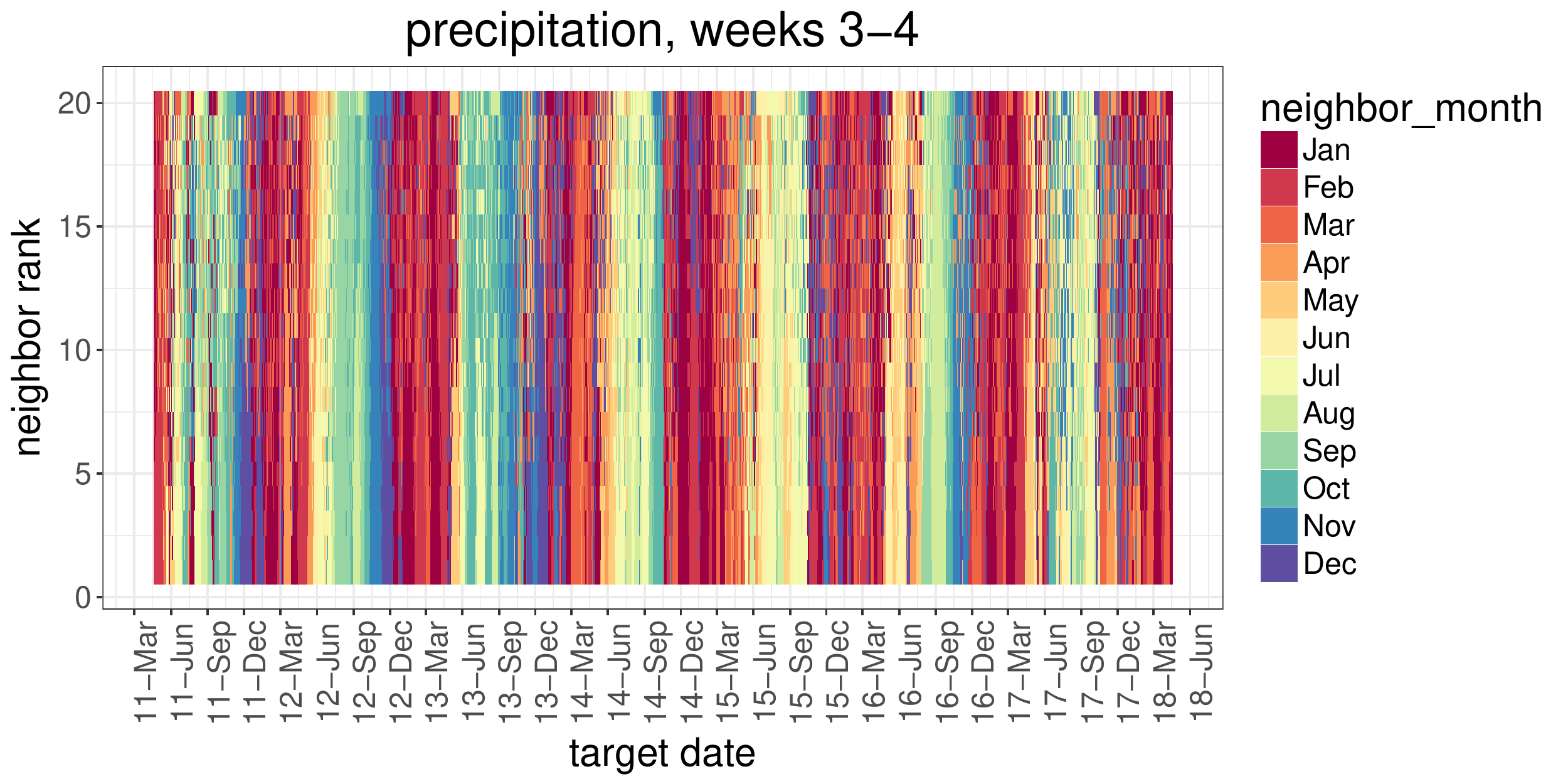}
\caption{Month of the 20 most similar neighbors learned by \autoknn (vertical axis ranges from $k=1$ to $20$) as a function of the target date to be predicted (horizontal axis). The plots for weeks 5-6 are very similar. Both temperature and precipitation neighbors are homogeneous in month for a given target date, but the months of the precipitation neighbors also exhibit a regular seasonal pattern from year to year, while the temperature neighbors do not.}
\label{fig:knn_matrix_plot}
\end{figure*}

\begin{figure*}[h!]
\includegraphics[width=\textwidth]{knn_matrix_plots/year-matplot_tmp2m-34w_maxnbrs20_all_dates-2011-2017.pdf}
\includegraphics[width=\textwidth]{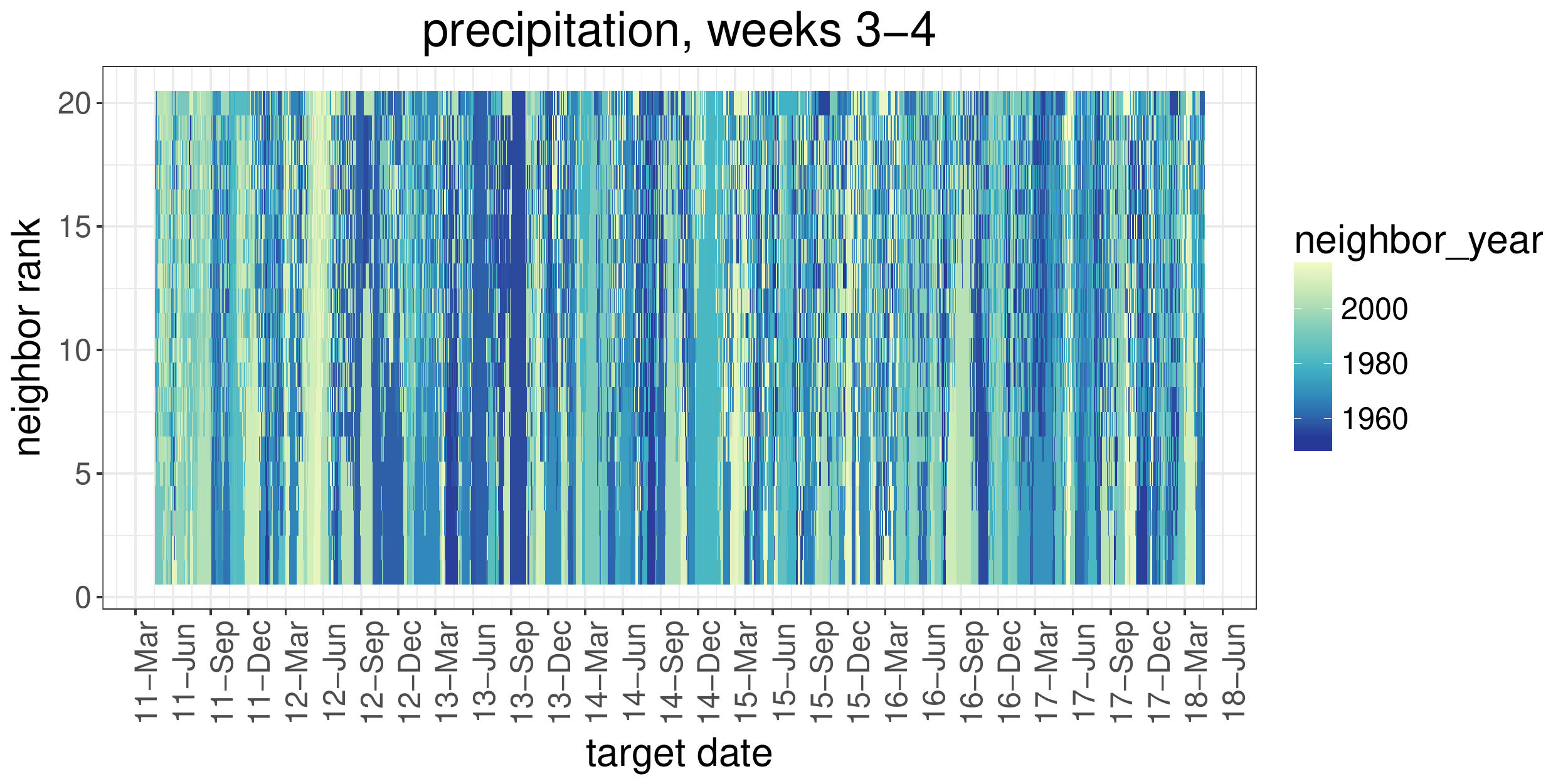}
\caption{Year of the 20 most similar neighbors learned by \autoknn (vertical axis ranges from $k=1$ to $20$) as a function of the target date to be predicted (horizontal axis). The plots for weeks 5-6 are very similar. The temperature neighbors are disproportionately drawn from recent years (post-2010), while the precipitation neighbors are not.}
\label{fig:knn_matrix_year_plot}
\end{figure*}}

\end{document}